\documentclass[11pt,letterpaper,onecolumn]{article}

\usepackage{graphicx,amssymb,epsfig,amssymb,amsmath,setspace,url,color}
\def\QEDclosed{\mbox{\rule[0pt]{1.3ex}{1.3ex}}} % for a filled box
\def\QED{\QEDclosed} % default to closed

\newsavebox{\savepar}

\newenvironment{boxitjournal}{\noindent\begin{lrbox}{\savepar}
\begin{minipage}[b]{6.2in}}
{\end{minipage}\end{lrbox}\fbox{\usebox{\savepar}}}

\newcommand{\av}{\mathbf{a}}
\newcommand{\bv}{\mathbf{b}}
\newcommand{\cv}{\mathbf{c}}

\newcommand{\ev}{\mathbf{e}}

\newcommand{\mv}{\mathbf{m}}
\newcommand{\nv}{\mathbf{n}}

\newcommand{\rv}{\mathbf{r}}
\newcommand{\sv}{\mathbf{s}}

\newcommand{\xv}{\mathbf{x}}
\newcommand{\yv}{\mathbf{y}}
\newcommand{\zv}{\mathbf{z}}

\newcommand{\Av}{\mathbf{A}}
\newcommand{\Bv}{\mathbf{B}}
\newcommand{\Cv}{\mathbf{C}}

\newcommand{\Ev}{\mathbf{E}}
\newcommand{\Fv}{\mathbf{F}}

\newcommand{\Iv}{\mathbf{I}}

\newcommand{\Mv}{\mathbf{M}}

\newcommand{\Rv}{\mathbf{R}}
\newcommand{\Sv}{\mathbf{S}}

\newcommand{\Wv}{\mathbf{W}}
\newcommand{\Xv}{\mathbf{X}}
\newcommand{\Yv}{\mathbf{Y}}

\newcommand{\onev}{\mathbf{1}}
\newcommand{\zerov}{\mathbf{0}}

\newcommand{\Sigmav}{\mathbf{\Sigma}}

\newcommand{\Muv}{\mathbf{\mathcal{M}}}

\newcommand{\betav}{\boldsymbol{\beta}}

\newcommand{\etav}{\boldsymbol{\eta}}

\newcommand{\muv}{\boldsymbol{\mu}}

\newcommand{\thetav}{\boldsymbol{\theta}}

\newcommand{\tr}{\mathop{\sf tr}}
\newcommand{\diag}{\mathop{\sf diag}}

\newcommand{\SNR}{\mbox{SNR}}
\newcommand{\LLR}{\mbox{LLR}}
\newcommand{\E}{\mathop{\sf E}}

\newcommand{\BEAS}{\begin{eqnarray*}}
\newcommand{\EEAS}{\end{eqnarray*}}
\newcommand{\BEA}{\begin{eqnarray}}
\newcommand{\EEA}{\end{eqnarray}}
\newcommand{\BEQ}{\begin{equation}}
\newcommand{\EEQ}{\end{equation}}
\newcommand{\BEQS}{\begin{equation*}}
\newcommand{\EEQS}{\end{equation*}}
\newcommand{\BIT}{\begin{itemize}}
\newcommand{\EIT}{\end{itemize}}
\newcommand{\BNUM}{\begin{enumerate}}
\newcommand{\ENUM}{\end{enumerate}}
\newcommand{\BIEA}{\begin{IEEEeqnarraybox*}}
\newcommand{\EIEA}{\end{IEEEeqnarraybox*}}
\newcommand{\BP}{\begin{proof}}
\newcommand{\EP}{\end{proof}}

\newcommand{\BA}{\begin{array}}
\newcommand{\EA}{\end{array}}
\newcommand{\BT}{\begin{tabular}}
\newcommand{\ET}{\end{tabular}}

\newcommand{\eg}{\mbox{e.g.}}
\newcommand{\ie}{\mbox{i.e.}}

\newcommand{\wrt}{\mbox{w.r.t.}}
\newcommand{\fig}{\mbox{Fig. }}

\definecolor{darkblue}{rgb}{0,0.25,0.5}
\definecolor{navy}{rgb}{0.2,0.3,0.8}
\definecolor{niceblue}{rgb}{0.04,0.52,0.78}
\definecolor{nicegreen}{rgb}{0.17,0.51,0.34}
\definecolor{truegreen}{rgb}{0,0.5,0}
\definecolor{havardred}{rgb}{0.6,0.2,0}
\definecolor{earthyellow}{rgb}{1,0.5,0}

\newtheorem{theorem}{Theorem}
\newtheorem{lemma}{Lemma}
\newtheorem{proposition}{Proposition}

\newtheorem{corollary}{Corollary}
\newtheorem{definition}{Definition}

\newcounter{exno}
{\begin{quote}\begin{small}\textbf{Example.}}%
{\end{small}\end{quote}}

{\begin{quote}\begin{small}\textbf{Examples.}}%
{\end{small}\end{quote}}

{\begin{quote}\begin{small}\textbf{Remark.}}%
{\end{small}\end{quote}}

{\begin{quote}}{\end{quote}}

\newcounter{oursection}

\newcounter{lecture}

\begin{document}

\title{\huge{A Variational Inference Framework for Soft-In-Soft-Out
Detection in Multiple Access Channels}}

\author{Darryl Dexu Lin and Teng Joon Lim
\thanks{D. D. Lin was with the Department of Electrical and Computer
Engineering, University of Toronto, 10 King's College Road, Toronto,
ON, Canada M5S 3G4. He is now with Qualcomm Incorporated, 5775 Morehouse Drive, San Diego, CA, 92121 USA (e-mail: darryl.lin@utoronto.ca).
T. J. Lim is with the Department of Electrical and Computer
Engineering, University of Toronto, 10 King's College Road, Toronto,
ON, Canada M5S 3G4 (e-mail: tj.lim@utoronto.ca). This work was
presented in part at the IEEE Globecom'05, St. Louis, MO, Nov.28 -
Dec.2, 2005. It was supported by a Natural Sciences and Engineering
Research Council (NSERC) of Canada Postgraduate Scholarship.}} \maketitle

\begin{abstract}
We propose a unified framework for deriving and studying
soft-in-soft-out (SISO) detection in interference channels using the
concept of variational inference. The proposed
framework may be used in multiple-access interference (MAI),
inter-symbol interference (ISI), and multiple-input multiple-output
(MIMO) channels. Without loss of generality, we will focus our
attention on turbo multiuser detection, to facilitate a more concrete discussion. It is
shown that, with some loss of optimality, variational inference avoids the exponential complexity
of \emph{a posteriori} probability (APP) detection by optimizing a
closely-related, but much more manageable, objective function called
\emph{variational free energy}. In addition to its systematic
appeal, there are several other advantages to this viewpoint. First
of all, it provides unified and rigorous justifications for numerous
detectors that were proposed on radically different grounds, and
facilitates convenient joint detection and decoding (utilizing the
turbo principle) when error-control codes are incorporated.
Secondly, efficient joint parameter estimation and data detection is
possible via the variational expectation maximization (EM)
algorithm, such that the detrimental effect of inaccurate channel
knowledge at the receiver may be dealt with systematically. We are also able to extend
BPSK-based SISO detection schemes to arbitrary square QAM
constellations in a rigorous manner using a variational argument.
\end{abstract}

% --------------------------------------------------------------------

\section{Introduction}

Following the discovery of turbo codes \cite{Berrou93}, the
principle of turbo processing has been used in various signal
processing settings. Among these, turbo detection for coded
transmission in interference channels, which treats the error
control code as the outer code and the interference channel as the
inner code, has been shown to perform dramatically better than the conventional non-iterative method of
interference suppression followed by hard-decision decoding.
Depending on the channel of interest, turbo detection includes
turbo multiuser detection for multiple access channels
\cite{Alex98,Wangx99}, turbo equalization for inter-symbol
interference (ISI) channels \cite{Doui95,Tuch02a}, and turbo MIMO
equalization for multiple-input multiple-output (MIMO) channels
\cite{Ariya00,Haykin04}. Due to the linear Gaussian vector channel
model that is common to these problems, techniques developed in
one area can often be readily applied to another with only minor
modifications. In this paper, we will restrict our signal model to
the multiuser detection (MUD) scenario. It should be understood
that the solutions proposed for this particular problem may be
generalized to turbo equalization and turbo MIMO detection settings as well.

The evolution of MUD research has seen detectors
being derived through many different approaches, such as the
minimization of mean-squared error (MMSE), decision-feedback, or
multistage interference cancellation (IC) \cite{Verdu98}. Within the
past decade, there has been a growing interest in coded CDMA
systems, where the need for joint detection and decoding leads to a
different class of multiuser detectors, namely turbo multiuser
detectors. Practical turbo multiuser detectors proposed in
\cite{Alex98} and \cite{Wangx99} are among the earliest and most
celebrated ones, due to their simplicity and remarkable performance.

Inside a turbo multiuser detector, a soft-in-soft-out (SISO)
detector component is of crucial importance, and is where the main design
challenges lie. It differs from the conventional detectors in that
it must be able to make use of prior knowledge of the symbols to be
detected, and the structure of the multiple access channel, to generate soft
symbol decisions. Unfortunately, unlike the decoder component, for
which feasible, low-complexity \emph{a posteriori} probability (APP)
generators (\eg, the BCJR algorithm \cite{Bahl74} for convolutional
codes) may be assumed, the optimal APP multiuser detector has
exponential complexity and is infeasible. As a result, suboptimal
SISO MUD design is key to the success of a practical turbo multiuser
detector.

In this paper, we intend to propose a generalized method for the
design of a SISO MUD, adopting a technique called \emph{variational
inference} \cite[pp. 422--436]{Mackay03}, which, like the
sum-product algorithm \cite{Ksch01}, is an approximate inference
algorithm in probabilistic models. We will see that this approach
not only successfully includes some important existing SISO MUD
schemes as special cases, but easily leads to various improvements
and extensions. Although our study focuses on SISO MUD by treating
it as an approximate inference engine, it also encompasses uncoded
MUD (detectors with no prior information and only hard decision
output), since uncoded MUD can be viewed as SISO MUD with
uniform prior distributions for the channel symbols.

Prior to this paper, recent attempts on providing a unified approach
to study the wide range of multiuser detectors include, to name a
few, \cite{Bout02}, \cite{Tanaka02} and \cite{Guo05}. Boutros and
Caire \cite{Bout02} generalize iterative multiuser joint decoding as
an approximate sum-product algorithm in a factor graph containing
both the multiuser channel and code constraints. Such a
generalization leads to elegant performance analysis through
\emph{density evolution}. Tanaka \cite{Tanaka02} and Guo and
Verd\'{u} \cite{Guo05} view the uncoded linear and optimal multiuser
detectors as \emph{posterior mean estimators} of the \emph{Bayes
retrochannel} such that, in the large system limit, the bit error
rate (BER) may be evaluated through techniques from statistical
physics. This paper may be regarded as an extension of
\cite{Tanaka02} and \cite{Guo05} into the realm of nonlinear (and
iterative) detectors. Specifically, we show that such detectors
arise from approximating the posterior distributions and iteratively
optimizing the approximate distributions, and address the
design challenges of the MUD component within the iterative
multiuser joint decoding problem, highlighted in \cite{Bout02}.

The implications of this new generalized framework are significant
in at least three ways:
\BNUM \item \textit{Theoretical Justification for Existing Multiuser
Detectors:} Section \ref{sec:variational} introduces the variational
inference formulation for MUD, in which a quantity known as
variational free energy is constructed and minimized, generating a
procedure termed variational free energy minimization (VFEM). From
this perspective, we will show how various uncoded linear multiuser
detectors (\eg, decorrelating and MMSE detectors), as well as their
interference cancellation extensions (\eg, unconstrained or clipped
successive interference cancellation (SIC) detectors) may be
derived. We will further argue that the VFEM approach naturally
produces SISO multiuser detectors that can be used in turbo MUD. In
particular, we will examine the celebrated algorithms proposed in
\cite{Alex98} and \cite{Wangx99}, to reveal that they can both
be derived with the VFEM approach.

\item \textit{Channel Parameter Joint Estimation Using Variational
EM Algorithm:} Section \ref{sec:var_EM} considers the scenario where
certain channel parameters are unknown or inaccurately estimated at
the multiuser receiver, motivating the joint estimation of channel
parameters together with unknown data symbols. The VFEM framework
offers a natural solution to this problem. By iteratively minimizing
the free energy over both the data symbols and the channel
parameters, we arrive at the \emph{variational EM} algorithm
\cite{Neal98}. This is a generalized EM algorithm with exact
inference in the E step replaced by variational inference. As
examples of this parameter estimation mechanism, we will demonstrate
how the unknown channel noise variance may be iteratively estimated,
and inaccurate channel amplitude refined, in conjunction with turbo
MUD.

\item \textit{Generalization of BPSK MUD to Square QAM
Modulation:} In bandwidth-constrained channels, extensions of the
SISO multiuser detectors from BPSK modulation to square QAM
modulation may also be carried out within the VFEM framework. These
extensions are not ad hoc, but optimal in the sense that the
variational free energy modified for $M$-QAM modulation is
minimized. Such a scheme gives rise to an iterative detection
technique for general linear Gaussian channels, called
\emph{Bit-Level Equalization and Soft Detection} (BLESD). It was
introduced in separate works of ours \cite{Lin06isit, Lin07isit}.

\ENUM

The rest of the paper will be organized as follows: Section
\ref{sec:system} describes the multiple access channel model and
formulates the optimal SISO multiuser detectors; Section
\ref{sec:message} discusses the decoding/detection scheduling issue
by studying the factor graph containing both the multiuser channel
and code constraints. This will prove to be an important design
parameter in the subsequent analysis of variational-inference-based
detectors. Sections \ref{sec:variational} and \ref{sec:var_EM}
contain the introduction and application examples of the proposed
variational inference framework for MUD, and in two directions (the
first two points summarized above) justify the merits of this new
point of view; Section \ref{sec:simulation1} presents some
simulation results, and Section \ref{sec:conclusion} concludes the
paper.

\emph{Notation}: Upper and lower case bold face letters indicate
matrices and column vectors, respectively; $\mathbf{1}$ represents
the all-one column vector; $\Xv \circ \Yv$ stands for the Schur
product (element-wise product) of matrices $\Xv$ and $\Yv$;
$\tr(\Xv)$ denotes the trace of a square matrix $\Xv$; $\diag(\xv)$
is a diagonal matrix with the vector $\xv$ on its diagonal;
$\diag(\Xv)$ is a diagonal matrix with the diagonal elements of
square matrix $\Xv$ on its diagonal; $\textsf{E}(\cdot)$ and
$\textsf{V}(\cdot)$ stand for the expected value and variance of a
random variable; $\mathcal{N}(\muv,\Sigmav)$ represents a Gaussian
pdf with mean $\muv$ and covariance matrix $\Sigmav$.

%%%%%%%%%%%%%%%%%%%%%%%%%%%%%%%%%%%%%%%%%%%%%%%%%%%%%%%%%%%%%%%
\section{System Description}\label{sec:system}

\subsection{Signal Model for BPSK Modulation}

Consider a synchronous DS-CDMA wireless link with $K$ users.
Assuming flat fading, by sampling the chip matched filter output at chip rate, the received
signal in one symbol interval, $\rv \in \mathbb{R}^{N \times 1}$,
can be written in the well-known vector form:
\begin{equation}\label{eq:CDMA}
    \rv=\Sv \Av \bv+\nv,
\end{equation}
where $\mathbf{S}=[\sv_1, \sv_2, \cdots, \sv_K]$ is the spreading
code matrix containing the normalized spreading sequences of the $K$
active users, $\Av = \diag(A_1,A_2,\cdots,A_K)$ is the channel
matrix representing each user's signal amplitude and $\bv=[b_1, b_2,
\cdots, b_K]^T$ contains the transmitted BPSK channel symbols from each
user. $\nv$ is a white Gaussian noise vector with distribution
$p(\nv)= \mathcal{N}(\zerov,\sigma^2 \Iv)$.

After bit-level matched filtering at the receiver, we may write
the matched filter output, $\yv \in \mathbb{R}^{K \times 1}$, as:
\begin{equation}\label{eq:MF_CDMA}
    \yv = \Sv^T \rv = \Rv \Av \bv + \zv,
\end{equation}
where $\Rv = \Sv^T\Sv$ is the symmetric normalized signature
correlation matrix with unit diagonal elements, and $\zv$ is a
coloured Gaussian noise vector with distribution $p(\zv)=
\mathcal{N}(\zerov,\sigma^2 \Rv)$.

The correlated noise statistics in $\yv$ may be whitened by
applying a noise whitening filter $\Fv^{-T}$, yielding
\BEQ\label{eq:DF_CDMA}
    \bar{\yv} = \Fv^{-T}\yv = \Fv\Av\bv + \bar{\nv},
\EEQ where $\Fv$ is a lower triangular matrix (\ie, $F_{ij} = 0$
for $i < j$) resulting from the Cholesky factorization for $\Rv$,
$\Rv = \Fv^T\Fv$. $\bar{\nv}$ is a white Gaussian noise vector,
having the same distribution as $\nv$.

As $\yv$ and $\bar{\yv}$ are sufficient statistics for detecting
$\bv$, equations (\ref{eq:CDMA}), (\ref{eq:MF_CDMA}) and
(\ref{eq:DF_CDMA}) are equivalent starting points for the derivation
of multiuser detectors, although certain computational savings are
easier to identify with certain models.

Note that the channel model for frequency selective and asynchronous
channels takes a similar linear form as (\ref{eq:CDMA}). Thus the
adaptation to these more general channel types is possible, but will
not be discussed explicitly here. Interested readers may refer to,
\eg, \cite{Wangx99}, for further insights.

\subsection{Optimal SISO Detectors}

Given the prior distribution $p(\bv)$ and the conditional
distribution $p(\rv|\bv)$, the jointly optimal (JO) detector uses Bayes rule to compute
\begin{equation}\label{eq:JO}
    p(\bv|\rv) = \frac{p(\rv|\bv)p(\bv)}{\sum_{\bv}
    p(\rv|\bv)p(\bv)}.
\end{equation}
The posterior distribution $p(\bv|\rv)$ is the ``soft output'' of
the jointly optimal detector; hard decisions are obtained by maximizing over all
possible symbol vectors $\bv$.

Similarly, the individually optimal (IO) detector is
obtained by evaluating the marginal posterior distribution of
$b_k$ ($k = 1$ to $K$):
\begin{equation}\label{eq:IO}
    p(b_k|\rv) = \frac{p(\rv|b_k)p(b_k)}{\sum_{b_k}
    p(\rv|b_k)p(b_k)},
\end{equation}
where $p(\rv|b_k)p(b_k) = \sum_{\bv \smallsetminus
b_k}p(\rv|\bv)p(\bv)$. Due to the discrete nature of the
information symbols, both jointly optimal and individually optimal detectors require prohibitive
exponential complexity.

The individually optimal detector is the optimal SISO multiuser detector in terms of
minimizing bit error rate (BER). Practical suboptimal SISO multiuser detectors may be
derived by taking in the \emph{prior} information $p(b_k)$ and
producing a \emph{posterior} probability $p(b_k|\rv)$ or
$p(b_k|\yv)$ through some intelligent approximation which does not
have exponential complexity. Variational inference is one example
of these ``intelligent approximations'', where the outcome,
$Q(b_k)$, which approximates $p(b_k|\rv)$, is found by optimizing an
underlying cost function called variational free energy, as will be
shown in Section \ref{sec:variational}.

%%%%%%%%%%%%%%%%%%%%%%%%%%%%%%%%%%%%%%%%%%%%%%%%%%%%%%%%%%%%%%%
\section{Message-Passing Scheduling in Turbo Multiuser Detection}\label{sec:message}

In a turbo multiuser detector, the detector section needs to be able
to accept prior estimates $\{p(b_k)\}_{k=1}^K$ from the APP decoder
and generate a soft decision, called \emph{extrinsic information}
(EXT), to be sent back to the APP decoder. Such a mechanism for EXT
exchange can be rigorously justified as the message passing
algorithm in graphs \cite{Ksch98,Bout02}. However, since any
practical multiuser detector is at best an \emph{approximation} to
the exact sum-product algorithm (because exact inference, with the
individually optimal detector, is NP complete), good methods to generate and pass EXT
are not unique.

In addition, the factor graph describing the statistical
dependencies among all unknowns (conditioned on the observations)
contains cycles, and hence several message passing schedules are
valid. In this section we describe the sequential, flooding and
hybrid schedules, and show that the Wang-Poor algorithm corresponds
to a hybrid scheduling. The
sequential schedule takes $K$ times as long as the flooding
schedule, but may result in fewer iterations to achieve a given
level of performance. While message-passing scheduling has not been thoroughly studied in the turbo MUD context, it is an important topic in iterative decoding of error control codes. For example, the different convergence rate of sequential and parallel (flooding) scheduling for decoding LDPC codes has been reported in \cite{Yacov07}.

%
%In this section, we will describe a few configurations of the
%interface between SISO multiuser detectors and APP decoders. In
%particular, we are interested in how to obtain EXT, and how to
%schedule the exchange of EXT. We will show that the scheduling
%mechanism often used in the literature (to be called the hybrid
%schedule) is in fact a convenient combination of two other more
%basic schemes, called sequential and flooding schedules. When the
%SISO detectors are viewed as variational inference engines, the
%validity of the newly proposed flooding schedule can be most easily
%seen. This novel approach can also be shown to substantially reduce
%the complexity and latency of the turbo detector.

\begin{figure}[hbt]
\centering	
\includegraphics[width=4.5in]{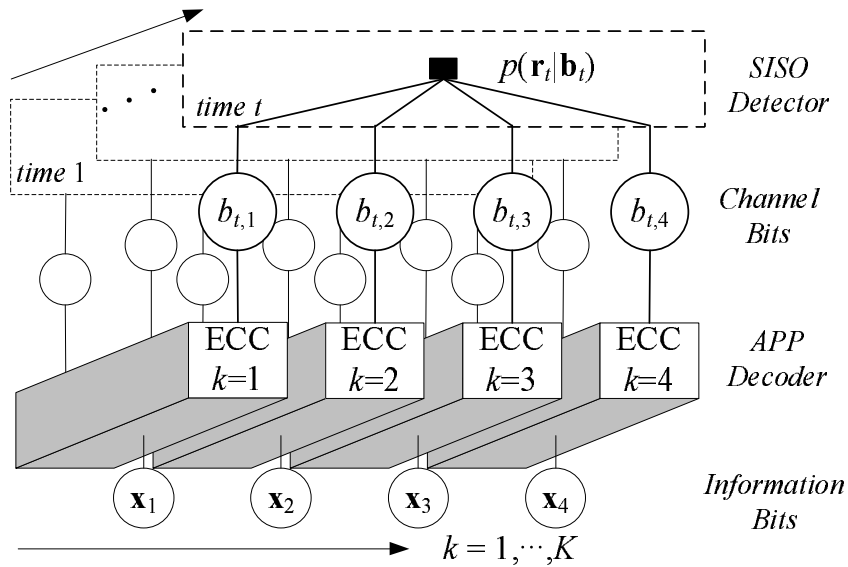}
\caption{Graphical model of a coded multiuser channel. Note the time
dependency among bits of the same user (code constraint), and the
user dependency among bits at the same time (channel constraint).}
\label{fig:message}
\end{figure}

From \fig \ref{fig:message}, it is seen that the nodes representing
the channel bits $\{b_{t,k}\}_{k=1}^K$ are the relay nodes that
separate the graph into two halves, where on one side the decoder
runs belief propagation to perform per-user APP decoding and on the
other side the multiuser detector performs variational inference.
The process by which the APP decoder retrieves prior information and
generates extrinsic information is standard (see \cite{Bahl74}) and
will be skipped. We will therefore only discuss
message passing between the detector and decoder.

\subsection{Obtaining Extrinsic Information: Sequential
Schedule}\label{sec:sequential}

\begin{figure}[hbt]
\centering
\includegraphics[width=5.5in]{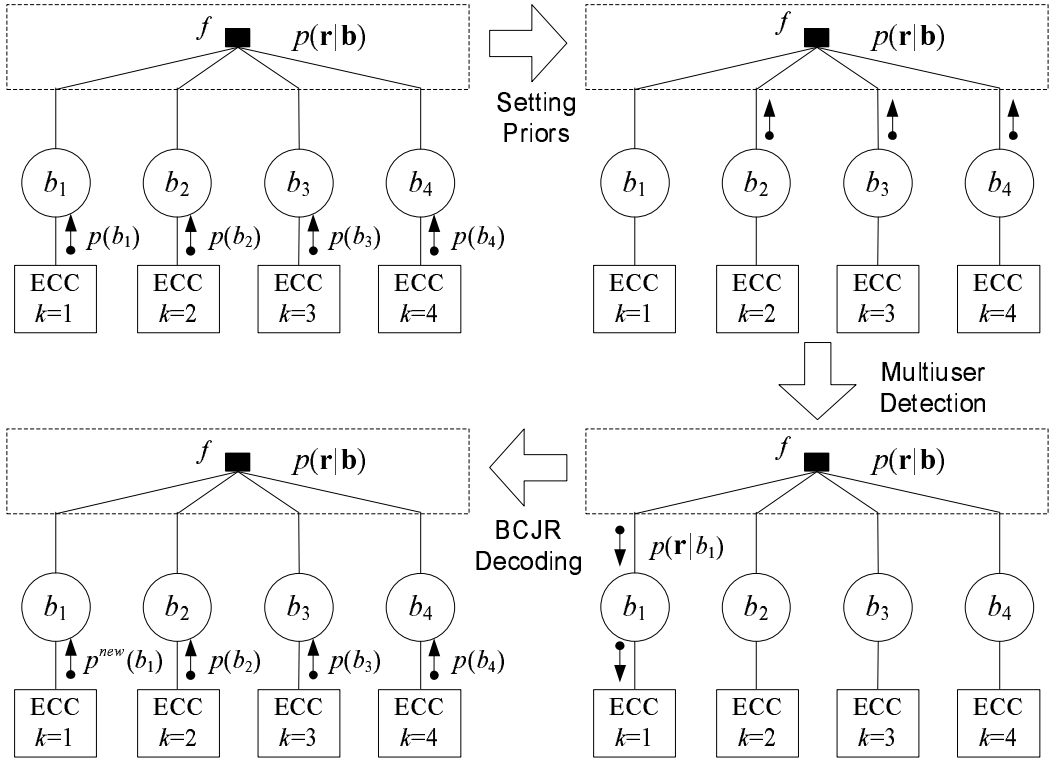}
\caption{An instance of sequential message-passing in the graphical
model: the multiuser detector receives prior distributions of $b_2$,
$b_3$ and $b_4$ to generate the extrinsic information for $b_1$.
This process is repeated for $b_2$, $b_3$ and $b_4$ to complete one
message-passing iteration.} \label{fig:sp_view}
\end{figure}

When a SISO detector is viewed as an approximate sum-product
algorithm \cite{Bout02}, the EXT may be obtained in a way analogous
to the message-passing rule in graphs. \fig \ref{fig:sp_view}
provides an example that demonstrates that the EXT for $b_1$ may be
generated using the priors of $b_2$, $b_3$ and $b_4$, but {\em not}
the prior of $b_1$.  In its exact form the message (EXT) from node
$f$ to node $b_1$ is \BEQ\label{eq:sp_message}
    \mathcal{M}_{f \rightarrow b_1} = \sum_{b_2,b_3,b_4} p(\rv|\bv) p(b_2)p(b_3)p(b_4) = p(\rv|b_1).
\EEQ In sequential scheduling, $\mathcal{M}_{f\rightarrow b_1}$ will
be passed into the APP decoder for user 1, which will generate a new
prior for $b_1$ that will be used for EXT generation for $b_2$, and
so on. So error control decoding is performed one user at a time,
and not in parallel.

In an {\em approximate} evaluation of EXT for $b_k$ that follows the
same vein, one would ignore the prior of $b_k$ even if it is
available from a previous iteration, and use a simple multi-user
detector such as linear MMSE to generate an estimated $p(\rv|b_k)$
using only $\{p(b_l)\}_{l \neq k}$. Thus in the sequential schedule,
\begin{itemize}
  \item the EXT for each bit is obtained using different inputs (prior distributions), necessitating a substantially different EXT generator (multiuser detector) for each bit; and
  \item the prior knowledge of $b_k$ is ignored before detection in generating the EXT for $b_k$.
\end{itemize}

The sequential schedule to obtain extrinsic information is
intuitive, since it resembles the message-passing protocol defined
in the sum-product algorithm \cite[ch. 4]{Jordan04}. But it is also
very restrictive, in that users have to be detected in series,
introducing latency in the detection process. Furthermore, since a
different joint detector must be devised for each user, the overall
complexity in general increases linearly with $K$ if no
simplification measures are taken.

%
%In practice, however, exact sum-product algorithm cannot be applied
%for large $K$, due to the exponential complexity associated with the
%summation in (\ref{eq:sp_message}). Various detections algorithms
%are used to approximate the sum-product approach (in other words,
%approximate $p(\rv|b_1)$), among which, \cite{Wangx99} has received the
%most attention. This algorithm will be studied in greater detail in
%Section \ref{sec:siso_mmse} to reveal its place within the general
%framework of variational-inference-based SISO detectors.

\subsection{Obtaining Extrinsic Information: Flooding
Schedule}\label{sec:flooding}

\begin{figure}[hbt]
\centering
\includegraphics[width=5.5in]{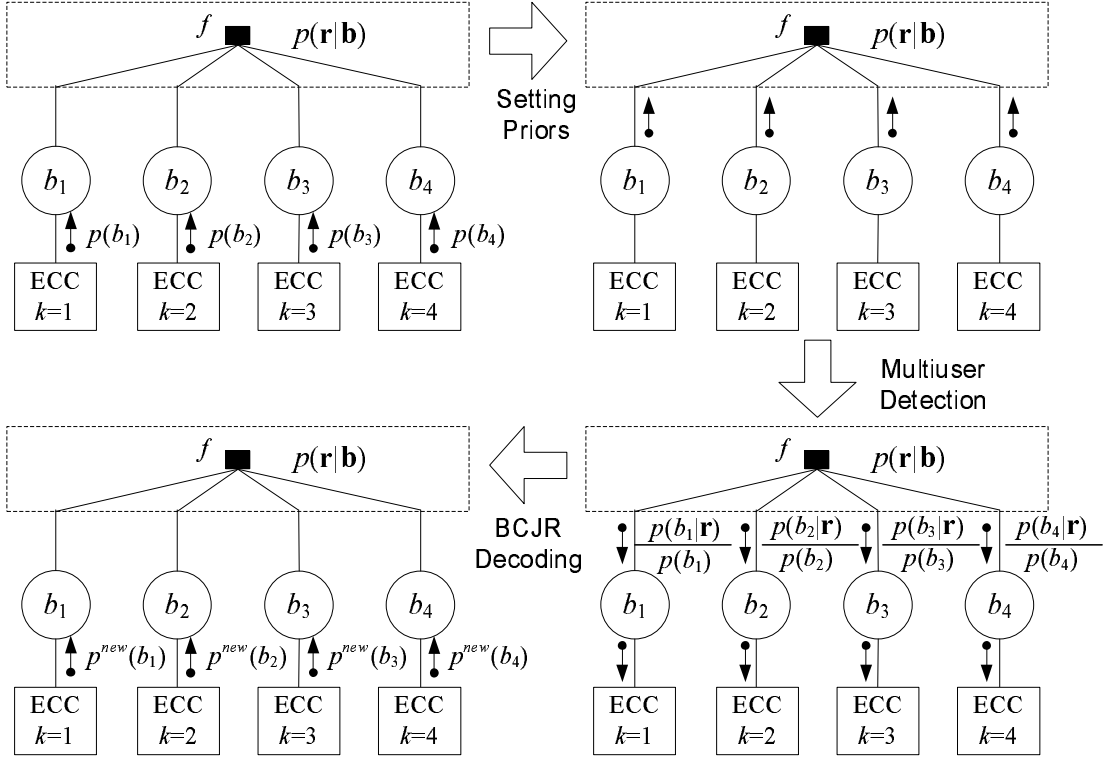}
\caption{An instance of flooding message-passing in the graphical
model: the multiuser detector receives prior distributions of
$b_1,\cdots,b_4$ to generate the extrinsic information for
$b_1,\cdots,b_4$. This completes one message-passing iteration.}
\label{fig:inf_view}
\end{figure}

In the flooding schedule, illustrated in \fig \ref{fig:inf_view},
EXT's for all bits are generated in parallel. The message from node
$f$ to $b_k$ will be \BEQ \label{eq:Mf_flood}
   \mathcal{M}_{f\rightarrow b_k} = \sum_{\{b_l\}_{l\neq k}} p(\rv|\bv)\prod_{l=1,l\neq k}^K p(b_l) \propto \frac{p(b_k|\rv)}{p(b_k)}.
\EEQ
Note that, unlike in sequential scheduling, all EXT's use the same
priors. For instance, $\mathcal{M}_{f\rightarrow b_2}$ and
$\mathcal{M}_{f\rightarrow b_4}$ both use $p(b_3)$ whereas in the
sequential schedule, $\mathcal{M}_{f\rightarrow b_4}$ would use
$p^{new}(b_3)$ from the most recent round of APP decoding.

As well, we can write the EXT of $b_k$ as
\BEQ
   \mathcal{M}_{f\rightarrow b_k} = \frac{1}{p(b_k)}\sum_{\{b_l\}_{l\neq k}} p(\rv|\bv)\prod_{l=1}^K p(b_l)
\EEQ and hence view the flooding schedule as making use of {\em all}
prior probabilities from the {\em same} iteration. This reasoning,
together with (\ref{eq:Mf_flood}), leads to the following
sub-optimal approximation:
\begin{itemize}
  \item Use all prior probabilities from the same iteration to generate an approximate $p(b_k|\rv)$, say $Q(b_k)$;
  \item Form the EXT for $b_k$ by dividing $Q(b_k)$ by $p(b_k)$;
  \item Send all EXT's to the $K$ APP decoders in parallel.
\end{itemize}
%
%The fundamental difference between the flooding schedule and the
%sequential schedule lies in the fact that, with flooding, the prior
%distributions of all users, \ie, $\{p(b_k)\}_{k=1}^K$ are used in
%detection. After the posterior distribution $p(b_k|\rv)$ is found by
%the multiuser detector, $p(b_k)$ is removed, such that
%\BEQ\label{eq:inf_message1}
%    \mathcal{M}_{f \rightarrow b_{t,k}} = \frac{p(b_k|\rv)}{p(b_k)} \propto p(\rv|b_k).
%\EEQ

%. The reason that the flooding schedule is also
%valid from the sum-product algorithm point of view can be seen by
%performing the regular sum-product algorithm after grouping the
%variable nodes $\{b_k\}_{k=1}^K$ into a ``supernode'' and invoking
%the \emph{mean field} approximation that $p(\bv|\rv) = \prod_{k=1}^K
%p(b_k|\rv)$.

The advantage of the flooding schedule is two fold: 1) By passing
messages to the detector in one shot, the latency is low; 2) By
generating the extrinsic information in one shot, the complexity of
the detector is reduced.

Through implementing the flooding schedule, our MUD design challenge
is shifted from approximating $p(\rv|b_k)$ to approximating
$p(b_k|\rv)$. And the variational inference viewpoint of MUD allows
us to easily do so.
%
%Therefore, with the multiuser detector seen as variational
%inference, the most natural way to obtain EXT is by evaluating the
%message as \BEQ\label{eq:inf_message2}
%    \mathcal{M}_{f \rightarrow b_{t,k}} \cong \frac{Q(b_{t,k})}{p(b_{t,k})}.
%\EEQ

\subsection{Obtaining Extrinsic Information: Hybrid
Schedule}\label{sec:hybrid}

A hybrid schedule can be defined in which the EXT for $b_k$ is
computed without using $p(b_k)$ like in the sequential schedule, and
all EXT's are computed in parallel like in the flooding schedule.
This approach removes the latency issue in sequential scheduling,
and has been used in the literature without justification.

%As we have indicated previously, the sequential schedule suffers
%from high latency and complexity. To avoid these deficiencies, the
%simple hybrid approach is widely adopted in the literature (although
%often not discussed in these terms). The philosophy is to apply the sequential
%schedule with one modification. That is, the APP decoders do not
%perform decoding until EXT's of all users are obtained. This
%simplification allows sequential decoding to be implemented in
%parallel, and the complexity of the multiuser detector to be
%reduced.

If exact inference is used to compute $p(\rv|b_k)$ in the hybrid
schedule, and $p(b_k|\rv)$ in the flooding schedule, the two
implementations are identical, since the messages coming out of the
MUD section are the same -- $\{p(\rv|b_k)\}_{k=1}^K$. However, in
practical detector design, $p(\rv|b_k)$ or $p(b_k|\rv)$ must be
approximated. As to be demonstrated in Section \ref{sec:siso_mmse},
$p(b_k|\rv)$ may be approximated as $Q(b_k)$ given prior
distributions $\{p(b_l)\}_{l=1}^K$, while $p(\rv|b_k)$ may be
approximated as $Q(b_k)$ given $\{p(b_l)\}_{l\neq k}$ and
non-informative $p(b_k)$. With these approximations, the hybrid and
flooding scheduling schemes differ, as the former becomes the
Wang-Poor turbo detector \cite{Wangx99} and the latter turns into a
brand-new design.

%The discussion in this section offers one other useful insight. It
%helps us recognize that very different inference techniques
%(variational inference for detection and belief propagation for
%decoding) can be linked using the turbo principle as the
%bridge, by exchange of messages between the two components. The
%establishment of a solid theory that allows the integration of
%variational inference algorithms in a graphical model is crucial,
%since variational inference is more suitable than belief propagation
%for many problems, such as the challenges of multiuser detection
%discussed in this paper.

%%%%%%%%%%%%%%%%%%%%%%%%%%%%%%%%%%%%%%%%%%%%%%%%%%%%%%%%%%%%%%%
\section{Multiuser Detection via Variational Inference}\label{sec:variational}

In \cite{Guo05}, Guo and Verd\'{u} treat the linear multiuser
detectors as \emph{posterior mean estimators} (PME) with
appropriately postulated distributions $p(\bv)$ and $p(\rv|\bv)$.
For example, if a Gaussian prior is assumed, \ie, $p(\bv) =
\mathcal{N}(\zerov,\Iv)$, and the channel is modelled as $p(\rv|\bv)
= \mathcal{N}(\Sv\Av\bv,\alpha^2 \Iv)$, the posterior (or
conditional) mean estimator, \ie, $\E\left[\bv|\rv\right]$, is a
generalized linear detector given by
\begin{equation}
    \hat{\bv} = \left(\Av^T\Sv^T \Sv\Av + \alpha^2 \Iv\right)^{-1}
    \Av^T\Sv^T \rv.
\end{equation}

By choosing different values for $\alpha$, we arrive at different
linear detectors. If $\alpha^2 = \sigma^2$, we get the MMSE
detector. If $\alpha \rightarrow 0$, we approach the decorrelating
detector. And if $\alpha \rightarrow \infty$, the matched filter
output is attained.

However, \cite{Guo05} has not considered another important class of detectors, namely the nonlinear detectors. In this work, we
wish to extend the coverage of the posterior mean estimator to nonlinear detectors by
introducing an additional degree of freedom in approximating the
posterior distribution. More specifically, we will not limit
ourselves to applying Bayes rule to calculate the posterior, but
instead use the more general and flexible \emph{variational
inference} technique.

\subsection{Variational Inference and Variational Free Energy
Minimization}\label{sec:var_inference}

We shall explain the variational inference method specifically in
terms of its application to multiuser detection, while a more
general and in-depth treatment, as well as its alternative
interpretations and connection to statistical physics, can be found
in \cite{Frey98}, \cite{Jordan99}, and \cite{Mackay03}.

As stated earlier, the general task of the SISO multiuser detector
is to perform inference on $\bv$ given the observation $\rv$, $\yv$
or $\bar{\yv}$ (we will simply use $\rv$ for now, as it is
understood that they are equivalent). Suppose our objective is the
jointly optimal detector, then the distribution of interest is\footnote{Strictly
speaking, individually optimal detector minimizes the BER. But since the difference is
minimal, we may consider the jointly optimal detector for simplicity}
$p(\bv|\rv)$. Very often, however, the direct evaluation of
$p(\bv|\rv)$ is computationally intractable when Bayes rule is
applied directly, in particular, when $p(\bv)$ is a discrete
distribution. In such a case, the variational inference technique
assumes a tractable approximation to $p(\bv|\rv)$, written as
$Q(\bv)$, where the constant $\rv$ is omitted for convenience.

A good approximation $Q(\bv)$ needs to resemble $p(\bv|\rv)$ as
closely as possible, and the \emph{Kullback-Leibler divergence} (or
\emph{relative entropy}) $\mathbb{D}\left[Q(\bv)\|p(\bv|\rv)\right]$ offers
an excellent measure of similarity. But since the distribution
$p(\bv|\rv)$ is difficult to attain as we have assumed, an
equivalent alternative, $p(\bv,\rv)=p(\rv|\bv)p(\bv)$, is used, and
$p(\bv,\rv)$ is called the \emph{complete likelihood function}. The
\emph{variational free energy} is thus defined as:
\begin{equation}\label{eq:free_F}
        \mathcal{F}(\lambda) = \int_{\bv} Q(\bv; \lambda)\log
    \frac{Q(\bv; \lambda)}{p(\bv,\rv)} d\bv,
\end{equation}
which equals $D\left[Q(\bv)\|p(\bv|\rv)\right]$ up to an additive
constant. In (\ref{eq:free_F}), $Q(\bv)$ is written as $Q(\bv; \lambda)$ to denote the dependence of $Q(\bv)$ on $\lambda$ explicitly, where $\lambda$ contains a set of
parameters that specify $Q(\bv)$. In the rest of the paper, we will however drop the dependence of the $Q$ function on $\lambda$ in accordance with the usual convention for writing probability distributions.

If no constraints are placed on $Q(\bv)$, by minimizing
$\mathcal{F}(\lambda)$, we reach $Q(\bv)=p(\bv|\rv)$ and nothing is
gained. But if we parameterize $Q(\bv)$ by assuming that it comes
from a restricted family of distributions (for example, a Gaussian),
we may very easily find a closed-form expression for
$\mathcal{F}(\lambda)$, which leads to a good approximation of
$p(\bv|\rv)$ via the minimization of variational free energy. This
method of performing approximate inference is called
\emph{variational inference}.

One important technique often used in variational inference is the
assumption that $Q(\bv)$ is factorizable as $\prod_{k=1}^K Q_k(b_k)$
(we shall omit the subscripts in $Q_k$ from here on), and find
distributions $\{Q(b_k)\}_{k=1}^K$ that minimize the free energy.
This factorization of a distribution and the independence assumption
associated with it is referred to as the \emph{mean-field
approximation} in statistical physics. A demonstration of its
application will be presented in detail in Section
\ref{sec:mean_field}.

The following is an outline of the general procedure for deriving
multiuser detectors through VFEM: \BNUM
\item \emph{Postulation}: Assume postulated distributions for $p(\bv)$, $p(\rv|\bv)$
and $Q(\bv)$; %
\item \emph{Evaluation}: Derive closed-form expression for
$\mathcal{F}(\lambda)$; %
\item \emph{Optimization}: Minimize $\mathcal{F}(\lambda)$
(exactly or iteratively) over $\lambda$. \ENUM

Note that we have now transformed the general MUD problem into a
well-defined optimization problem, with a unique objective function,
called variational free energy. This procedure bears close
resemblance to the routine of deriving thermodynamic state equations
in statistical physics \cite{Path04}, which is not surprising,
considering the fact that variational inference is indeed rooted in
statistical physics.

\subsection{VFEM Interpretation of Linear Multiuser
Detectors}\label{sec:linear_MUD}

We shall begin by deriving linear multiuser detectors from
variational free energy minimization, and thus show that simply
adjusting the postulated distributions $p(\bv)$, $p(\rv|\bv)$ and
$Q(\bv)$ leads to the well-known decorrelating and MMSE detectors.
Although the exercises presented here are somewhat trivial, since
uncoded linear MUD is the simplest instance of MUD, they lay
the foundation for more sophisticated variations in later sections.\\

\begin{proposition}\label{ca:decorrelating} \emph{Decorrelating
Detectors} may be derived through the VFEM routine by assuming the
following distributions:
    \begin{equation}\label{eq:postulate_decorrelating}
    \left\{\begin{array}{rcl}
    p(\bv) &=& Constant\\
    p(\rv|\bv) &=& \mathcal{N}(\Sv\Av\bv,\sigma^2 \Iv)\\
    Q(\bv) &=& \mathcal{N}(\muv, \Sigmav).
    \end{array} \right.
    \end{equation}
\end{proposition}

\begin{proof}
    Evaluating $\mathcal{F}(\lambda)$ as in (\ref{eq:free_F}), we have
    a function of $\muv$ and $\Sigmav$:
    \begin{equation}\label{eq:decorrelating_F}
    \mathcal{F}(\muv,\Sigmav) = -\frac{1}{2} \log|\Sigmav| + \frac{1}{2\sigma^2} \left\{ \muv^T\Av^T\Sv^T \Sv\Av\muv  + \tr[(\Av^T\Sv^T \Sv\Av) \Sigmav] - 2\rv^T \Sv \Av \muv \right\}
    \end{equation}
    The final estimate of $Q(\bv)$ is given by the minimizers
    $\hat{\muv}$ and $\hat{\Sigmav}$ of $\mathcal{F}(\muv,\Sigmav)$. Calculating $\partial \mathcal{F}(\muv)/ \partial \muv$ and $\partial \mathcal{F}(\Sigmav)/ \partial
    \Sigmav^{-1}$ and equating to zero, we have
    \begin{equation} \label{eq:hatmuv_decorrelating}
    \begin{array}{rcl}
        \hat{\muv} &=& (\Av^T\Sv^T \Sv\Av)^{-1} \Av^T \Sv^T \rv\\
        \hat{\Sigmav} &=& \sigma^2 (\Av^T\Sv^T \Sv\Av)^{-1} .
    \end{array}
    \end{equation}
If hard decisions are desired, $\hat{\muv}$ can be used as the
detector output, since it maximizes $Q(\bv)$, which is Gaussian. It
is easy to recognize that $\hat{\muv}$ is identical to the
decorrelating detector output.
\end{proof}

Note that given the postulated priors in
(\ref{eq:postulate_decorrelating}), the exact posterior $p(\bv|\rv)$
is tractable and is in fact Gaussian. Therefore, the solved $Q$
function, $\mathcal{N}(\hat{\muv},\hat{\Sigmav})$, is the exact
posterior distribution which could also have been found by applying
Bayes rule directly.

The decorrelating detector uses non-informative priors for the data
bits transmitted, by setting $p(\bv)$ to a constant. But in
practice, side information is available. For instance,
$\{b_k\}_{k=1}^K$ can be safely assumed to be i.i.d. and zero mean.
For BPSK signaling, in particular, we also known that
$\textsf{E}(b_k^2) = 1$. We will subsequently show that the Gaussian
approximation about $p(\bv)$, utilizing the first and second order
statistics of $\bv$, gives rise to the familiar MMSE detector.

\begin{proposition}\label{ca:linear} \emph{MMSE Multiuser
Detectors} may be derived through the VFEM routine by assuming the
following distributions:
    \begin{equation}\label{eq:postulate_linear}
    \left\{ \begin{array}{rcl}
    p(\bv) &=& \mathcal{N}(\zerov, \Iv)\\
    p(\rv|\bv) &=& \mathcal{N}(\Sv\Av\bv,\sigma^2 \Iv)\\
    Q(\bv) &=& \mathcal{N}(\muv, \Sigmav).
    \end{array} \right.
    \end{equation}
\end{proposition}

\begin{proof}
    Evaluating $\mathcal{F}(\lambda)$ yields a function of $\muv$ and $\Sigmav$:
    \begin{equation}\label{eq:linear_F}
    \mathcal{F}(\muv,\Sigmav) = -\frac{1}{2} \log|\Sigmav| +
    \frac{1}{2\sigma^2} \left\{ \muv^T(\Av^T\Sv^T \Sv\Av + \sigma^2 \Iv)\muv + \tr[(\Av^T\Sv^T \Sv\Av + \sigma^2 \Iv)\Sigmav] - 2\rv^T \Sv \Av \muv \right\}
    \end{equation}
    Solving $\partial \mathcal{F}(\muv)/ \partial \muv = \zerov$ and $\partial \mathcal{F}(\Sigmav)/ \partial
    \Sigmav^{-1} = \zerov$ leads to the following solution:
    \begin{equation} \label{eq:hatmuv_mmse}
    \begin{array}{rcl}
        \hat{\muv} &=& (\Av^T\Sv^T \Sv\Av + \sigma^2 \Iv)^{-1} \Av^T \Sv^T \rv \\
        \hat{\Sigmav} &=& \sigma^2 (\Av^T\Sv^T \Sv\Av + \sigma^2
        \Iv)^{-1}.
    \end{array}
    \end{equation}
Apparently, $\hat{\muv}$ in (\ref{eq:hatmuv_mmse}) can be identified
as the MMSE detector output.
\end{proof}

Note that the variational inference interpretation of decorrelating
and MMSE detectors also produces a covariance matrix $\Sigmav$ of
the $Q$ function, which is not available through conventional signal
processing techniques. $\Sigmav$ indicates the reliability of the
detector output, something the hard-decision detector is unable to
make use of. But it will prove valuable in SISO detectors, as
demonstrated in Sections \ref{sec:siso_mmse} and
\ref{sec:mean_field}.

\subsection{VFEM Interpretation of Interference Cancellation
Detectors}\label{sec:sic_mud}

Iterative multiuser detectors, and especially their convergence
behavior, have been actively researched in the past. In
\cite{Elder98}, linear SIC and PIC are categorized as the
Gauss-Seidel and Jacobi iterations for solving linear equations. SIC
is also analyzed in greater depth in \cite{Ras00} and \cite{Guo00}.
The study is later extended to clipped SIC in \cite{Tan01} through
the investigation of the variational inequality (VI) problem. Here
we offer an alternative view of SIC as the coordinate descent
algorithm applied to the minimization of
$\mathcal{F}(\muv,\Sigmav)$.

\begin{proposition}\label{ca:iterative} \emph{Linear/Clipped SIC
Detectors} may be derived from assuming the same distributions as in
(\ref{eq:postulate_decorrelating}) or (\ref{eq:postulate_linear}),
except by minimizing $\mathcal{F}(\muv,\Sigmav)$ using the
\emph{coordinate descent algorithm}. That is, in the $i$-th
iteration, for $k=1$ to $K$,
    \begin{equation}
    \BA{cl}
    \min_{\mu_k} &
    \mathcal{F}(\mu_1^{(i)},\cdots,\mu_{k-1}^{(i)},\mu_k,\mu_{k+1}^{(i-1)},\cdots,\mu_K^{(i-1)},\Sigmav)\\
    $s.t.$ & \mu_{min} \leq \mu_k \leq \mu_{max}.
    \EA\end{equation}
    The algorithm describes a
    linear SIC if $\mu_{min}=-\infty$ and $\mu_{max}=\infty$, and
    a clipped SIC otherwise.
\end{proposition}

\begin{proof}
    Setting $\partial \mathcal{F}(\muv,\Sigmav)/ \partial \mu_k = 0$
    based on (\ref{eq:linear_F}) yields
    \begin{equation}\label{eq:linear_SIC_eq}
        A_k\sv_k^T \Sv\Av\muv - A_k\sv_k^T \rv + \sigma^2 \mu_k = 0.
    \end{equation}
    Rearranging the terms and defining $\muv_{\setminus k} =
    [\mu_1,\cdots,\mu_{k-1},0,\mu_{k+1},\cdots,\mu_K]^T$, the optimal
    $\mu_k$ is then expressed in the familiar linear interference
    cancellation form if $\mu_k$ is unbounded (\ie, $\mu_{min}=-\infty$ and $\mu_{max}=\infty$):
    \begin{equation}\label{eq:linear_SIC}
        \hat{\mu}_k = \frac{1}{A_k^2 + \sigma^2}
        A_k\sv_k^T(\rv-\Sv\Av\muv_{\setminus k}).
    \end{equation}
    Since updating $\mu_k$ ($k=1,\cdots,K$) consecutively subject to $\partial \mathcal{F}(\muv,\Sigmav)/ \partial \mu_k =
    0$ is the coordinate descent algorithm for
    minimizing $\mathcal{F}(\muv,\Sigmav)$, then (\ref{eq:linear_SIC}) corresponds to the coordinate descent implementation of the
    MMSE detector. On the other hand, setting $\partial \mathcal{F}(\muv,\Sigmav)/ \partial \mu_k = 0$
    based on (\ref{eq:decorrelating_F}) leads to the the coordinate descent implementation of the decorrelating detector:
    \begin{equation}\label{eq:linear_SIC1}
        \hat{\mu}_k = \frac{1}{A_k}\sv_k^T(\rv-\Sv\Av\muv_{\setminus
        k}),
    \end{equation}
    which is the standard-form SIC detector seen in the literature. If $\mu_{min}$ and $\mu_{max}$ are
    finite, we need to solve (\ref{eq:linear_SIC_eq}) subject
    to $\mu_{min} \leq \mu_k \leq \mu_{max}$, which corresponds to clipped SIC.
\end{proof}

To verify that (\ref{eq:linear_SIC}) and (\ref{eq:linear_SIC1}) do
converge to MMSE or decorrelator solutions, and to gain further
insights into the convergence behavior when the optimization
constraints are active (clipped SIC), we invoke the following
theorem \cite{Luo92}:

\begin{theorem}[Luo and Tseng, 1992]
\it{
    Consider an optimization problem:
    \begin{equation}\label{eq:theorem}
    \begin{array}{ll}
    min\ f(\xv)=g(\Ev \xv) + \cv^T\xv, &s.t.\ \xv \in \mathcal{X},
    \end{array}
    \end{equation}
    where $\mathcal{X}$ is a box (possibly unbounded) in
    $\mathbb{R}^n$, $f$ is a proper closed convex function in
    $\mathbb{R}^n$, $g$ is a proper closed convex function in
    $\mathbb{R}^m$, $\Ev$ is an $m\times n$ matrix having no zero
    column, and $\cv \in \mathbb{R}^n$.
    Also assume
    \begin{enumerate}
    \item The set of optimal solutions for (\ref{eq:theorem}),
    denoted by $\mathcal{X}^*$ is nonempty;
    \item The domain of $g$ is open and $g$ is strictly convex twice
    continuously differentiable on the domain;
    \item $\nabla^2 g(\Ev\xv^*)$ is positive definite for all $\xv^* \in
    \mathcal{X}^*$.
    \end{enumerate}

    Then if $\{\xv^r\}$ is a sequence of iterates generated by coordinate
    descent method according to the Almost Cyclic Rule or
    Gauss-Southwell Rule, $\{\xv^r\}$ converges at least linearly
    to an element of $\mathcal{X}^*$.
}
\end{theorem}

Since the objective function of optimization,
$\mathcal{F}(\muv,\Sigmav)$, satisfies all conditions in the theorem
when the spreading codes are linearly independent, it is clear that
this theorem applies to the general linear/clipped SIC setting. Also
due to the objective function being quadratic and the constraints
being linear, there is a unique optimal solution in $\mathcal{X}^*$.
We may thus conclude the following:

\begin{corollary}
\emph{Linear/Clipped SIC} are guaranteed to converge to the unique
minimum free energy defined by $\mathcal{F}(\muv,\Sigmav)$ and the
constraint ($\mu_{min} \leq \mu_k \leq \mu_{max}$), and the rate of
convergence is at least linear.
\end{corollary}

This result is proven for the first time to our knowledge.
Additionally, we may relax the conventional cyclic order of
iteration for SIC and assert that as long as the coordinates are
iterated upon according to either the \emph{Almost Cyclic Rule} or
\emph{Gauss-Southwell Rule}, at least linear convergence rate is
guaranteed. These relaxed iteration rules are discussed in
\cite{Luo92}.

In the sequel, we will investigate a few SISO multiuser detectors
within the variational inference framework. Unlike the uncoded
detectors studied previously, we will now make use of the soft
output provided by $Q(\bv)$ to facilitate iterative multiuser joint
decoding. We will demonstrate that a unique SISO detector is
determined by choosing 1) the postulated distributions (like
(\ref{eq:postulate_decorrelating}) and (\ref{eq:postulate_linear}),
but with biased priors), and 2) the message-passing schedule for
joint decoding.

\subsection{VFEM Interpretation of Gaussian SISO Multiuser
Detector}\label{sec:siso_mmse}

\begin{boxitjournal}
\begin{definition}
A \emph{Gaussian SISO Multiuser Detector} is a multiuser detector
that obtains soft estimates $Q(\bv)$ through the VFEM routine,
subject to the following postulated distributions:
    \begin{equation}\label{eq:siso_mud_postulate}
    \left\{\begin{array}{rcl}
    p(\bv) &=& \mathcal{N}(\tilde{\bv}, \Wv)\\
    p(\rv|\bv) &=& \mathcal{N}(\Sv\Av\bv, \sigma^2 \Iv)\\
    Q(\bv) &=& \mathcal{N}(\muv, \Sigmav),
    \end{array} \right.
    \end{equation}
where $\tilde{\bv} = [\tilde{b}_1,\cdots,\tilde{b}_K]^T$ are the
soft bit estimates from the APP decoder, and $\Wv =
\diag([1-\tilde{b}_1^2,\cdots,1-\tilde{b}_K^2]^T)$.
\end{definition}
\end{boxitjournal}

We name this detector \emph{Gaussian SISO MUD} because, like the
uncoded linear detectors in Section \ref{sec:linear_MUD}, Gaussian
densities are assumed for the prior and posterior distributions of
$\bv$. But unlike the linear detectors, this detector is capable of
accepting informative priors, as well as generating soft posterior
bit probability.

\subsubsection{The Existing Form of Gaussian SISO
MUD}\label{sec:exist_gaussian_siso}

A ground-breaking turbo detection scheme was proposed by Wang and
Poor \cite{Wangx99}, spurring a tremendous amount of interest in
turbo MUD and turbo equalization in the years that followed. It
involves a two stage process: First, the soft bit estimate from the
APP decoder is remodulated and subtracted from the matched filter
output: \BEQ\label{eq:yv_k}
    \yv_k \triangleq \yv - \Rv\Av\tilde{\bv}_k,
\EEQ
In (\ref{eq:yv_k}), $\tilde{\bv}_k =
[\tilde{b}_1,\cdots,\tilde{b}_{k-1},0,\tilde{b}_{k+1},\cdots,\tilde{b}_K]^T$,
which is equal to the soft bit estimates coming from the APP
decoder, $\tilde{\bv}$, except for the $k$-th element being $0$.

Second, a linear MMSE filter is used to further suppress the
residual interference. It can be shown that the filter output is
\BEQ\label{eq:z_k}
    z_k = A_k\ev_k^T\overbrace{\left[\Av^T\Wv_k\Av + \sigma^2\Rv^{-1}\right]^{-1}}^{\mbox{MMSE with residual MAI}} \overbrace{ \left[\Rv^{-1}\yv-\Av\tilde{\bv}_k\right]}^{\mbox{Soft IC}},
\EEQ where $\ev_k$ denotes a $K$-vector of all zeros, except for the
$k$-th element being $1$, and $\Wv_k =
\diag([1-\tilde{b}_1^2,\cdots,1-\tilde{b}_{k-1}^2,1,1-\tilde{b}_{k+1}^2,\cdots,1-\tilde{b}_K^2]^T)$.

In order to convert the MMSE filter output $z_k$ into a soft
estimate in the discrete domain, a Gaussian equivalent channel
assumption is made about $z_k$, \ie,
\BEQ
    z_k = \alpha_k b_k + \eta_k,
\EEQ where $\alpha_k$ is a constant and $p(\eta_k) =
\mathcal{N}(0,\nu_k^2)$. In other words, $p(z_k|b_k) =
\mathcal{N}(\alpha_k b_k, \nu_k^2)$. Since $\alpha_k$ and $\nu_k^2$
can be found to be, respectively, \BEQ\label{eq:alpha_nu}\BA{rcl}
    \alpha_k &=& A_k^2 \left[(\Av^T\Wv_k\Av + \sigma^2 \Rv^{-1})^{-1} \right]_{k,k}\\
    \nu_k^2 &=& z_k - z_k^2, \EA\EEQ
the output EXT can be written as \BEQ
    \LLR_{mud}(b_k) = \log \frac{p(\rv|b_k=1)}{p(\rv|b_k=-1)} \thickapprox \log \frac{p(z_k|b_k = 1)}{p(z_k|b_k = -1)} =
    \frac{2 z_k}{1-\alpha_k}.
\EEQ

In essence, the target distribution $p(\rv|b_k)$ is approximated by
$p(z_k|b_k)$ to obtain the EXT. We will now demonstrate that with
the VFEM formulation, the two-stage process can be derived from a
single optimization procedure, and without the heuristic Gaussian
assumption about $z_k$.

\begin{proposition}
    The SISO multiuser detection scheme described in \cite{Wangx99}
    is an instance of \emph{Gaussian SISO MUD}.
\end{proposition}

\BP
If the extrinsic information is extracted following the sequential
schedule in Section \ref{sec:sequential}, by ignoring the prior
information for $b_k$, then (\ref{eq:siso_mud_postulate}) may be
modified as
    \begin{equation}\label{eq:siso_mud_postulate1}
    \left\{\begin{array}{rcl}
    p(\bv) &=& \mathcal{N}(\tilde{\bv}_k, \Wv_k)\\
    p(\rv|\bv) &=& \mathcal{N}(\Sv\Av\bv, \sigma^2 \Iv)\\
    Q(\bv) &=& \mathcal{N}(\muv_k, \Sigmav_k),
    \end{array} \right.
    \end{equation}
where $\tilde{\bv}_k =
[\tilde{b}_1,\cdots,\tilde{b}_{k-1},0,\tilde{b}_{k+1},\cdots,\tilde{b}_K]^T$
and $\Wv_k =
\diag([1-\tilde{b}_1^2,\cdots,1-\tilde{b}_{k-1}^2,1,1-\tilde{b}_{k+1}^2,\cdots,1-\tilde{b}_K^2]^T)$.
From (\ref{eq:siso_mud_postulate1}), it can be shown that
\begin{equation}\label{eq:F_siso}
\begin{array}{rcl}
    \mathcal{F}_{gauss}(\muv_k,\Sigmav_k) &=& \frac{1}{2\sigma^2} [ \muv_k^T(\Av^T\Sv^T\Sv\Av+\sigma^2\Wv_k^{-1})\muv_k - 2(\rv^T\Sv\Av +\sigma^2\tilde{\bv}_k^T\Wv_k^{-1})\muv_k]\\
    && -\frac{1}{2}\log|\Sigmav_k| + \frac{1}{2}\tr(\Wv_k^{-1}\Sigmav_k) + \frac{1}{2\sigma^2}\tr(\Av^T\Sv^T\Sv\Av\Sigmav_k).
\end{array}
\end{equation}
Let $\mu_k'$ denote the $k$-th element of $\muv_k$. Solving
$\partial \mathcal{F}_{gauss}(\muv_k,\Sigmav_k)/\partial \mu_k' = 0$
yields
\begin{equation}\label{eq:mu}
    \mu_k' =
    \ev_k^T(\Av^T\Sv^T\Sv\Av+\sigma^2\Wv_k^{-1})^{-1}\Av^T\Sv^T(\rv-\Sv\Av\tilde{\bv}_k),
\end{equation}
which is identical to $z_k$ in (\ref{eq:z_k}).

One piece of information that the MMSE-based detector in
\cite{Wangx99} does not have is the covariance matrix of the
posterior distribution, $\Sigmav_k$, which can be shown to be
\BEQ\label{eq:Sigma}
    \Sigmav_k = \left(\frac{1}{\sigma^2}\Av^T\Sv^T\Sv\Av +
    \Wv_k^{-1}\right)^{-1}.
\EEQ

In other words, the marginal posterior distribution of $b_k$ is
$Q(b_k) = \mathcal{N}(\mu_k', [\Sigmav_k]_{k,k})$. Since the prior
distribution of $b_k$ is ignored during the detection operation,
$Q(b_k)$ obtained as such is in fact proportional to $p(\rv|b_k)$.
Therefore, \BEQ
    \LLR_{mud}(b_k) = \log \frac{p(\rv|b_k=1)}{p(\rv|b_k=-1)} \thickapprox \log \frac{Q(b_k=1)}{Q(b_k=-1)} =
    \frac{2\mu_k'}{[\Sigmav_k]_{k,k}}.
\EEQ

Applying the matrix inversion lemma on $\Sigmav_k$ in
(\ref{eq:Sigma}), we have \BEQ
    \Sigmav_k = \Wv_k - \Wv_k \Av(\Av\Wv_k\Av + \sigma^2 \Rv^{-1} )^{-1} \Av
    \Wv_k.
\EEQ
Since $[\Wv_k]_{k,k} = 1$, $[\Sigmav_k]_{k,k} = 1-
A_k^2[(\Av\Wv_k\Av + \sigma^2 \Rv^{-1})^{-1}]_{k,k} = 1 - \alpha_k$,
where $\alpha_k$ is as defined in (\ref{eq:alpha_nu}). Therefore,
\BEQ
    \LLR_{mud}(b_k) = \frac{2\mu_k'}{[\Sigmav_k]_{k,k}} = \frac{2 z_k}{1-\alpha_k}.
\EEQ
We have thus re-derived the Wang-Poor scheme via a radically
different approach. It is remarkable how the variational inference
viewpoint leads to exactly the same outcome as \cite{Wangx99}, while
the conditional Gaussian assumption made about the MMSE filter
output is no longer necessary.
\EP

After taking APP decoding into account, the Wang-Poor turbo MUD
algorithm as a whole can be seen as hybrid-Gaussian-SISO MUD. In the
next section, we will systematically investigate all three possible
scheduling schemes applied to Gaussian SISO MUD.

\subsubsection{The Standard Forms of Gaussian SISO
MUD}\label{sec:stand_gaussian_siso}

\begin{table}
\caption{Three scheduling schemes of turbo MUD employing Gaussian
SISO MUD.}\label{tab:gauss_siso} \small
\begin{center}
\setlength{\tabcolsep}{3pt}
\begin{tabular}{|lllcl|}  \hline
\multicolumn{5}{|c|}{\textbf{Sequential-Gaussian-SISO}} \\
\hline
\multicolumn{5}{|l|}{\emph{Initialization}: $\tilde{\bv} = \zerov$} \\
\multicolumn{5}{|l|}{FOR $j = 1:J$ (\emph{Outer Iteration})}\\
    &\multicolumn{4}{l|}{FOR $k=1:K$} \\
    &   & $\tilde{\bv}_k$ &=& $[\tilde{b}_1,\cdots,\tilde{b}_{k-1},0,\tilde{b}_{k+1},\cdots,\tilde{b}_K]^T$\\
    &   & $\Wv_k$ &=&
            $\diag([1-\tilde{b}_1^2,\cdots,1-\tilde{b}_{k-1}^2,1,1-\tilde{b}_{k+1}^2,\cdots,1-\tilde{b}_K^2]^T)$\\
    &   & $\mu_k'$ &=& $A_k\ev_k^T\left[\Av^T\Wv_k\Av + \sigma^2\Rv^{-1}\right]^{-1} \left[\Rv^{-1}\yv-\Av\tilde{\bv}_k\right]$\\
    &   & $\alpha_k$ &=& $A_k^2 \left[(\Av^T\Wv_k\Av + \sigma^2 \Rv^{-1})^{-1} \right]_{k,k}$ \\
    &   & \multicolumn{3}{l|}{$\LLR_{mud}(b_k) = \frac{2\mu_k'}{1-\alpha_k}$} \\
    &   & \multicolumn{3}{l|}{$\LLR_{dec}(b_k) \stackrel{\scriptsize \mbox{Decoding}}{\Longleftarrow} \LLR_{mud}(b_k)$} \\
    &   & $\tilde{b}_k$ &=& $\tanh[\LLR_{dec}(b_k)/2]$\\
    &END& & &\\
END& & & &\\ \hline \hline

\multicolumn{5}{|c|}{\textbf{Flooding-Gaussian-SISO}} \\
\hline
\multicolumn{5}{|l|}{\emph{Initialization}: $\tilde{\bv} = \zerov$} \\
\multicolumn{5}{|l|}{FOR $j = 1:J$ (\emph{Outer Iteration})}\\
    &\multicolumn{4}{l|}{FOR $k=1:K$} \\
    &   & $\tilde{\bv}_k$ &=& $[\tilde{b}_1,\cdots,\tilde{b}_{k-1},0,\tilde{b}_{k+1},\cdots,\tilde{b}_K]^T$\\
    &   & $\Wv$ &=&
            $\diag([1-\tilde{b}_1^2,\cdots,1-\tilde{b}_K^2]^T)$\\
    &   & $\check{\mu}_k$ &=& $A_k\ev_k^T\left[\Av^T\Wv\Av + \sigma^2\Rv^{-1}\right]^{-1} \left[\Rv^{-1}\yv-\Av\tilde{\bv}_k\right]$\\
    &   & $\check{\alpha}_k$ &=& $(1-\tilde{b}_k^2)A_k^2 \left[(\Av^T\Wv\Av + \sigma^2 \Rv^{-1})^{-1} \right]_{k,k}$ \\
    &   & \multicolumn{3}{l|}{$\LLR_{mud}(b_k) = \frac{2\check{\mu}_k}{1-\check{\alpha}_k}$} \\
    &END& & &\\

    &\multicolumn{4}{l|}{FOR $k=1:K$} \\
    &   & \multicolumn{3}{l|}{$\LLR_{dec}(b_k) \stackrel{\scriptsize \mbox{Decoding}}{\Longleftarrow} \LLR_{mud}(b_k)$} \\
    &   & $\tilde{b}_k$ &=& $\tanh[\LLR_{dec}(b_k)/2]$\\
    &END& & &\\
END& & & &\\
\hline\hline

\multicolumn{5}{|c|}{\textbf{Hybrid-Gaussian-SISO}} \\
\hline
\multicolumn{5}{|l|}{\emph{Initialization}: $\tilde{\bv} = \zerov$} \\
\multicolumn{5}{|l|}{FOR $j = 1:J$ (\emph{Outer Iteration})}\\
    &\multicolumn{4}{l|}{FOR $k=1:K$} \\
    &   & $\tilde{\bv}_k$ &=& $[\tilde{b}_1,\cdots,\tilde{b}_{k-1},0,\tilde{b}_{k+1},\cdots,\tilde{b}_K]^T$\\
    &   & $\Wv_k$ &=&
            $\diag([1-\tilde{b}_1^2,\cdots,1-\tilde{b}_{k-1}^2,1,1-\tilde{b}_{k+1}^2,\cdots,1-\tilde{b}_K^2]^T)$\\
    &   & $\mu_k'$ &=& $A_k\ev_k^T\left[\Av^T\Wv_k\Av + \sigma^2\Rv^{-1}\right]^{-1} \left[\Rv^{-1}\yv-\Av\tilde{\bv}_k\right]$\\
    &   & $\alpha_k$ &=& $A_k^2 \left[(\Av^T\Wv_k\Av + \sigma^2 \Rv^{-1})^{-1} \right]_{k,k}$ \\
    &   & \multicolumn{3}{l|}{$\LLR_{mud}(b_k) = \frac{2\mu_k'}{1-\alpha_k}$} \\
    &END& & &\\

    &\multicolumn{4}{l|}{FOR $k=1:K$} \\
    &   & \multicolumn{3}{l|}{$\LLR_{dec}(b_k) \stackrel{\scriptsize \mbox{Decoding}}{\Longleftarrow} \LLR_{mud}(b_k)$} \\
    &   & $\tilde{b}_k$ &=& $\tanh[\LLR_{dec}(b_k)/2]$\\
    &END& & &\\
END& & & &\\
\hline

\end{tabular}
\setlength{\tabcolsep}{5pt}
\end{center}
\normalsize
\end{table}

In Table \ref{tab:disc_siso}, we summarize three different versions
of standard Gaussian SISO MUD. In the following, we point out some
of the major characteristics associated with each one, and, in
particular, introduce the new flooding schedule implementation.

\emph{Sequential-Gaussian-SISO:} In Section
\ref{sec:exist_gaussian_siso}, we presented an elegant
variational-inference-based approach to obtain the EXT at the SISO
detector output, which coincides with the EXT conventionally
calculated through soft interference cancellation and MMSE
filtering.

In contrast to \cite{Wangx99}, however, where the EXT's are stored
until all users are processed and then used for APP decoding in
parallel, the sequential schedule requires the EXT,
$\LLR_{mud}(b_k)$, be directly passed down to the APP decoder. Then
the EXT from the APP decoder, viewed by the detector as the updated
prior $\tilde{b}_k$, is immediately used for the detection of
$b_{k+1}$.

\emph{Flooding-Gaussian-SISO:} The flooding schedule allows the APP
decoding of all users to be done in parallel. In the detection
stage, some changes to the derivation presented in Section
\ref{sec:exist_gaussian_siso} are needed, since the prior
information of $b_k$ should not be ignored as is done in
(\ref{eq:siso_mud_postulate1}). Instead, the postulated
distributions in (\ref{eq:siso_mud_postulate}) are adopted.

Given (\ref{eq:siso_mud_postulate}), the free energy becomes
\begin{equation}\label{eq:F_siso1}
\begin{array}{rcl}
    \mathcal{F}_{gauss}(\muv,\Sigmav) &=& \frac{1}{2\sigma^2} [ \muv^T(\Av^T\Sv^T\Sv\Av+\sigma^2\Wv^{-1})\muv - 2(\rv^T\Sv\Av +\sigma^2\tilde{\bv}^T\Wv^{-1})\muv]\\
    && -\frac{1}{2}\log|\Sigmav| + \frac{1}{2}\tr(\Wv^{-1}\Sigmav) + \frac{1}{2\sigma^2}\tr(\Av^T\Sv^T\Sv\Av\Sigmav).
\end{array}
\end{equation}

Solving $\partial \mathcal{F}(\muv,\Sigmav)/\partial \muv = \zerov$
and $\partial \mathcal{F}(\muv,\Sigmav)/\partial \Sigmav^{-1} =
\zerov$ leads to the minimizer of
$\mathcal{F}_{gauss}(\muv,\Sigmav)$ in (\ref{eq:F_siso}):
\BEQ\label{eq:mu_optimal}\BA{rcl}
    \muv &=& \tilde{\bv} +
    (\Av^T\Sv^T\Sv\Av+\sigma^2\Wv^{-1})^{-1}\Av^T\Sv^T(\rv-\Sv\Av\tilde{\bv})\\
    \Sigmav &=& (\sigma^{-2}\Av^T\Sv^T\Sv\Av+\Wv^{-1})^{-1}.
\EA\EEQ

It implies that the approximate posterior distribution, $p(\bv|\rv)
\thickapprox Q(\bv) = \mathcal{N}(\muv, \Sigmav)$. In other words,
the marginal posterior distribution of $b_k$ is $p(b_k|\rv)
\thickapprox \mathcal{N}(\mu_k, [\Sigmav]_{k,k})$. Recalling in
(\ref{eq:siso_mud_postulate}), $p(b_k) = \mathcal{N}(\tilde{b}_k,
1-\tilde{b}_k^2)$, if we apply the flooding schedule in Section
\ref{sec:flooding} to extract the EXT, then \BEQ
    p(\rv|b_k) \varpropto \frac{p(b_k|\rv)}{p(b_k)} \thickapprox \frac{Q(b_k)}{p(b_k)}=
    \mathcal{N}(\mu_{ext},\sigma_{ext}^2),
\EEQ where \BEQ\label{eq:mu_sigma}\BA{rcl}
    \mu_{ext} &=& \sigma_{ext}^2\left(\frac{\mu_k}{[\Sigmav]_{k,k}} - \frac{\tilde{b}_k}{1-\tilde{b}_k^2}\right) \\
    \frac{1}{\sigma_{ext}^2} &=& \frac{1}{[\Sigmav]_{k,k}} - \frac{1}{1-\tilde{b}_k^2}.
\EA\EEQ
(\ref{eq:mu_sigma}) is true, because if
$\mathcal{N}(\mu_1,\sigma_1^2)\mathcal{N}(\mu_2,\sigma_2^2)
\varpropto \mathcal{N}(\mu_3,\sigma_3^2)$, then \cite{Ksch01}
\BEQ\BA{rcl}
    \frac{\mu_3}{\sigma_3^2} &=& \frac{\mu_1}{\sigma_1^2} +
    \frac{\mu_2}{\sigma_2^2}\\
    \frac{1}{\sigma_3^2} &=& \frac{1}{\sigma_1^2} +
    \frac{1}{\sigma_2^2}.
\EA\EEQ

Finally, sampling $\mathcal{N}(\mu_{ext},\sigma_{ext}^2)$ at $b_k =
1$ and $b_k=-1$, we obtain
\BEQ\displaystyle\BA{rcl}
    \LLR_{mud}(b_k) &=& \log \frac{p(\rv|b_k=1)}{p(\rv|b_k=-1)}\\
    &\thickapprox& \frac{2\mu_{ext}}{\sigma_{ext}^2}\\
    &=& \frac{2 \mu_k}{[\Sigmav]_{k,k}} - \frac{2 \tilde{b}_k}{1-\tilde{b}_k^2}\\
    &=& \frac{2 \ev_k^T[\tilde{\bv} + \Wv\Av(\Av\Wv\Av + \sigma^2\Rv^{-1})^{-1}(\Rv^{-1}\yv - \Av\tilde{\bv})]}{(1-\tilde{b}_k^2) - (1-\tilde{b}_k^2)^2 A_k^2\left[(\Av\Wv\Av + \sigma^2\Rv^{-1})\right]_{k,k}} - \frac{2
    \tilde{b}_k}{1-\tilde{b}_k^2}\\
    &=& \frac{2 \ev_k^T[\Wv\Av(\Av\Wv\Av + \sigma^2\Rv^{-1})^{-1}(\Rv^{-1}\yv - \Av\tilde{\bv}_k)] - 2\tilde{b}_k (1-\tilde{b}_k^2) A_k^2\left[(\Av\Wv\Av + \sigma^2\Rv^{-1})\right]_{k,k}}{(1-\tilde{b}_k^2)\left\{1 - (1-\tilde{b}_k^2) A_k^2\left[(\Av\Wv\Av +
    \sigma^2\Rv^{-1})^{-1}\right]_{k,k}\right\}}\\
    &=& \frac{2(1-\tilde{b}_k^2)A_k\ev_k^T (\Av\Wv\Av + \sigma^2\Rv^{-1})^{-1}(\Rv^{-1}\yv - \Av\tilde{\bv}_k)}{(1-\tilde{b}_k^2)\left\{1 - (1-\tilde{b}_k^2) A_k^2\left[(\Av\Wv\Av +
    \sigma^2\Rv^{-1})^{-1}\right]_{k,k}\right\}}\\
    &=& \frac{2\check{\mu}_k}{1 - \check{\alpha}_k^2},
\EA\EEQ
where \BEQ\label{eq:check_mu_alpha}\BA{rcl}
    \check{\mu}_k &=& A_k \ev_k^T(\Av\Wv\Av + \sigma^2\Rv^{-1})^{-1}(\Rv^{-1}\yv -
    \Av\tilde{\bv}_k)\\
    \check{\alpha}_k &=& (1-\tilde{b}_k^2) A_k^2\left[(\Av\Wv\Av +
    \sigma^2\Rv^{-1})^{-1}\right]_{k,k}.
\EA\EEQ
In (\ref{eq:check_mu_alpha}), $\check{\mu}_k$ can also be computed
more efficiently as \BEQ
    \check{\mu}_k = A_k \ev_k^T(\Av\Wv\Av + \sigma^2\Rv^{-1})^{-1}(\Rv^{-1}\yv -
    \Av\tilde{\bv}) + \tilde{b}_k A_k^2 \left[(\Av\Wv\Av +
    \sigma^2\Rv^{-1})^{-1}\right]_{k,k},
\EEQ
such that common information may be utilized to evaluate
$\check{\mu}_k$ for all $k$.

\emph{Hybrid-Gaussian-SISO:} As mentioned earlier, the Wang-Poor
turbo MUD scheme is exactly the hybrid-Gaussian-SISO MUD. It differs
from the sequential schedule in that the EXT for $b_k$ generated by
the SISO detector is now stored until the EXT's of all users
$k=1,\cdots,K$ are ready. Then EXT's are passed down to the APP
decoders, for decoding in parallel. Hybrid-Gaussian-SISO MUD brings
computational savings compared to the more optimal
sequential-Gaussian-SISO MUD, due to both the possibility of
parallel decoding, and the ease of evaluating $[\Av^T\Wv_k\Av +
\sigma^2 \Rv^{-1}]^{-1}$.

%In Table \ref{tab:gauss_siso}, we summarize both the
%sequential-schedule and flooding-schedule Gaussian SISO MUD to
%highlight the underlying distinctions between the two. The
%hybrid-schedule approach is easily obtained from the
%sequential-schedule approach.

So far, based on the Gaussian distributions assumed in the
postulation step, we showed that the variational inference algorithm
converges to a family of Gaussian SISO detectors, including the
well-known Wang-Poor scheme as the special case. But the VFEM
framework allows us to generalize even further, since the Gaussian
distributions, albeit convenient, are unnatural choices for BPSK
symbols. The subsequent section will focus on a different family of
detectors induced by a different set of assumptions in the
postulation step.

\subsection{VFEM Interpretation of Discrete SISO Multiuser
Detector}\label{sec:mean_field}

\begin{boxitjournal}
\begin{definition}
A \emph{Discrete SISO Multiuser Detector} is a multiuser detector
that obtains soft estimates $Q(\bv)$ through the VFEM routine,
subject to the following postulated distributions:
    \begin{equation}\label{eq:mean_field_postulate}
    \left\{\begin{array}{rcl}
    p(\bv) &=& \prod_{k=1}^K
    \xi_k^{\frac{1+b_k}{2}}(1-\xi_k)^{\frac{1-b_k}{2}}, \ \ b_k
    \in \{\pm 1\}\\
    p(\rv|\bv) &=& \mathcal{N}(\Sv\Av\bv,\sigma^2 \Iv)\\
    Q(\bv) &=& \prod_{k=1}^K
    \gamma_k^{\frac{1+b_k}{2}}(1-\gamma_k)^{\frac{1-b_k}{2}}, \ \ b_k
    \in \{\pm 1\},
    \end{array}\right.
    \end{equation}
where $\xi_k$ and $\gamma_k$ are the prior and posterior probability
of $b_k$ being $1$.
\end{definition}
\end{boxitjournal}

The \emph{discrete SISO MUD} has two salient features in the
postulated distributions: 1) Both the prior and posterior
distributions are discrete, conforming to the actual properties of
the data; 2) The posterior distributions of individual bits,
$\{b_k\}_{k=1}^K$, are assumed to be independent by applying the
\emph{mean-field} approximation. Indeed, the only distinction
between this scheme and the jointly optimal detector is the mean-field
approximation about the posterior, which, though a crude assumption
in general, is asymptotically exact in the large system limit. This
technique is closely tied to the replica method used to study the
performance of randomly spread CDMA \cite{Guo05}. The mean-field
approximation is also used in \cite{Fabri02} and \cite{Kaba03} to
derive multiuser detectors for uncoded CDMA.

\subsubsection{The Existing Form of Discrete SISO
MUD}\label{sec:exist_discrete_siso}

In \cite{Alex98}, a simple (linear complexity) multiuser detector
was proposed for coded CDMA producing near optimal performance at
very high network load. Alexander, Grant and Reed applied a simple
interference cancellation scheme and made the following observation:
\begin{equation}\label{eq:Alex_eq}
\begin{array}{rl}
    & p(y_k|b_k,\bv_{\setminus k} = \tilde{\bv}_k) \\
    =& \frac{1}{\sqrt{2\pi\sigma^2}}\exp\left\{ -\frac{1}{2\sigma^2}(y_k - \sv_k^T\Sv\Av\tilde{\bv}_k-A_k^2b_k)^2\right\}\\
    =& \frac{1}{\sqrt{2\pi\sigma^2}}\exp\left\{ -\frac{1}{2\sigma^2}[A_k b_k - \sv_k^T(\rv - \Sv\Av\tilde{\bv}_k)]^2\right\}
\end{array}
\end{equation}
where $\tilde{\bv}$ is the average bit estimate received from the
APP decoder, and $\tilde{\bv}_k =
[\tilde{b}_1,\cdots,\tilde{b}_{k-1},0,\tilde{b}_{k+1},\cdots,\tilde{b}_K]^T$.
Defining $\sigma_{tot}^2 = \sigma^2 + \sigma_{MU}^2$ as the variance
of the combined channel noise and residual MAI modelled as Gaussian
noise, $\sigma_{tot}^2$ can be approximated as the sample average of
$[\sv_k^T(\rv - \Sv\Av\tilde{\bv})]^2$.

The soft estimate of $b_k$ can then be drawn from (\ref{eq:Alex_eq})
as a log-likelihood ratio: \BEQ\label{eq:exist_LLR}
    \LLR(b_k) =  \frac{2}{\sigma_{tot}^2} A_k\sv_k^T(\rv - \Sv\Av\tilde{\bv}_k),
\EEQ%
and fed back to the APP channel code decoders. The decoders
subsequently update $\tilde{\bv}$ for the next iteration. Now we
proceed to prove the link of this simple and effective scheme to the
VFEM framework.

\begin{proposition}
    The SISO multiuser detection scheme described in \cite{Alex98}
    is an instance of the \emph{Discrete SISO MUD}.
\end{proposition}

\BP
Let the prior distribution $p(\bv)$ in
(\ref{eq:mean_field_postulate}) represent the EXT provided by the APP decoder. Also, $p(\bv)=\prod_{k=1}^K p(b_k)$, where $p(b_k)
= \xi_k^{\frac{1+b_k}{2}}(1-\xi_k)^{\frac{1-b_k}{2}}$ implies that
$p(b_k=1) = (\xi_k)^{1}(1-\xi_k)^{0} = \xi_k$ and $p(b_k=0) =
(\xi_k)^{0}(1-\xi_k)^{1} = 1- \xi_k$. As seen from the derivation in
(\ref{eq:Alex_eq}), in the traditional MUD viewpoint, this
information may be used for soft interference cancellation in the
detection stage. We will now demonstrate that this IC technique
corresponds to one iteration of recursive minimization of
variational free energy.

We let $\tilde{b}_k = 2\xi_k -1$ and $m_k = 2\gamma_k - 1$, to
denote the prior mean and posterior mean of $b_k$. After some
mathematical manipulation, we have, according to
(\ref{eq:mean_field_postulate}) and (\ref{eq:free_F}),
\begin{equation}\label{eq:f_meanfield}
\begin{array}{rcl}
    \mathcal{F}_{disc}(\mv) &=& \sum_{k=1}^K \left[\frac{1+m_k}{2} \log \frac{1+m_k}{1+\tilde{b}_k} +
    \frac{1-m_k}{2}
    \log\frac{1-m_k}{1-\tilde{b}_k}\right] + \frac{N}{2}\log \sigma^2\\
    && + \frac{1}{2\sigma^2}\left[\rv^T\rv - \rv^T\Sv\Av\mv + \mv^T\Bv\mv + \tr(\Av^T\Sv^T\Sv\Av)\right],
\end{array}
\end{equation}
where $\Bv = \Av^T\Sv^T\Sv\Av - \diag{(\Av^T\Sv^T\Sv\Av)}$.
(\ref{eq:f_meanfield}) is obtained by utilizing the property that
\BEQ\BA{rcl}
    \E(\bv^T \Cv \bv) &=& \E\left(\sum_{i\neq j}C_{ij}b_i b_j + \sum_{i=1}^K C_{ii}b_i^2 \right)\\
&=& \mv^T [\Cv-\diag(\Cv)]\mv + \onev^T\diag(\Cv)\onev, \EA\EEQ
for $\bv\in \{\pm 1\}^K$ and $\Cv = [C_{ij}] \in \mathbb{R}^{K\times
K}$.

Rearranging $\partial \mathcal{F}_{disc}(\mv)/\partial \mv = \zerov$
gives a system of equations, for $k=1,\cdots,K$, that determines the
minimum of $\mathcal{F}_{disc}(\mv)$,
\begin{equation}\label{eq:mean_gamma}
    \log\frac{1+m_k}{1-m_k} = \log\frac{1+\tilde{b}_k}{1-\tilde{b}_k} +
    \frac{2}{\sigma^2}\left[\etav_k^T \rv - \betav_k^T\mv\right],
\end{equation}
where $\etav_k$ and $\betav_k$ are the $k$-th column vectors of
$\Sv\Av$ and $\Bv$, respectively. The coordinate descent algorithm
minimizes a function successively along one direction at a time. By
setting $\partial \mathcal{F}_{disc}(\mv)/\partial m_k$ to zero in
turn, we have the following update for user $k$ in iteration $i$:
\begin{equation} \label{eq:LLRupdate}
   \LLR^{(i)}(b_k) = \LLR^{(0)}(b_k) + \frac{2}{\sigma^2}\left[\etav_k^T \rv - \betav_k^T\mv_{<k}^{(i)} - \betav_k^T\mv_{>k}^{(i-1)}\right].
\end{equation}

In (\ref{eq:LLRupdate}), we defined the log-likelihood ratio
$\LLR^{(i)}(b_k) \triangleq \log\frac{1+m_k^{(i)}}{1-m_k^{(i)}}$ (or
equivalently, $m_k^{(i)} = \tanh[\LLR^{(i)}(b_k)/2]$). The
iterations are initialized with the prior probabilities of $b_k$,
\ie, $\mv^{(0)} = \tilde{\bv}$ and $\LLR^{(0)}(b_k) =
\log\frac{1+\tilde{b}_k}{1-\tilde{b}_k}$. As well, $\mv_{<k} =
[m_1,\cdots,m_{k-1},\overbrace{0,\cdots,0}^{K-k+1}]^T$, while
$\mv_{>k} = [\overbrace{0,\cdots,0}^{k-1},m_k,\cdots,m_K]^T$.

The flooding schedule (see \fig \ref{fig:inf_view}) indicates that
the EXT is the ratio between the posterior and the prior
distributions, or the difference between posterior and prior
$\LLR$'s, \ie, after $I$ iterations, the multiuser detector passes
the following EXT to the $k$-th decoder:
\BEQ\label{eq:disc_LLR}
    \LLR_{mud}(b_k) = \LLR_{pos}(b_k) - \LLR^{(0)}(b_k) = \frac{2}{\sigma^2} \left[\etav_k^T \rv - \betav_k^T\mv_{<k}^{(I)} - \betav_k^T\mv_{>k}^{(I-1)}\right].
\EEQ

Consider simplifying (\ref{eq:disc_LLR}) by removing the serial
iterations, then \BEQ
    \LLR_{mud}(b_k) = \frac{2}{\sigma^2} \left[\etav_k^T \rv -
    \betav_k^T\mv^{(0)}\right] = \frac{2}{\sigma^2} A_k\sv_k^T(\rv - \Sv\Av\tilde{\bv}_k),
\EEQ which is similar to (\ref{eq:exist_LLR}). Note that this
simplified updating scheme does not guarantee the decrease of free
energy, and thus is not as robust as the standard version in
(\ref{eq:disc_LLR}).
\EP

In the above proof, we have set $\sigma^2$ to be the channel noise
variance, and assumed it known. This is in contrast to
Alexander-Grant-Reed's original derivation, where $\sigma_{tot}^2$
is the noise-plus-MAI variance which has to be estimated
iteratively. We will postpone the discussion of this issue until
Section \ref{sec:em_meanfield}, where we will show that, the
iterative estimation of $\sigma_{tot}^2$ can be interpreted as the M
step in the variational EM algorithm for joint data detection and
noise variance estimation.

Also from (\ref{eq:LLRupdate}), an interesting link to uncoded
multi-stage SIC can be made -- in that case, $\LLR^{(0)}(b_k) = 0$.
Defining $\hat{b}_k^{(i)} = m_k^{(i)} = \tanh[\LLR^{(i)}(b_k)/2]$,
we get the hyperbolic-tangent SIC updates
\begin{equation}
   \hat{b}_k^{(i)} = \tanh\left\{\frac{1}{\sigma^2} A_k\sv_k^T\left(\rv - \Sv\Av\hat{\bv}_{<k}^{(i)} - \Sv\Av\hat{\bv}_{>k}^{(i-1)}\right)\right\}.
\end{equation}

%\emph{Remark 1 (Relation to Hard-Decision SIC)}: It is easily seen
%that hard decision SIC is a special case of the Discrete SISO MUD
%described by (\ref{eq:hypertan}), by replacing the $\tanh(\cdot)$
%function with a slicer. It is generally true that uncoded
%hard-decision MUD algorithms can be drawn from SISO MUD, by 1)
%using non-informative priors and 2) making hard-decision on the
%soft-decision output.

In addition to demonstrating a solid theoretical foundation for the
Alexander-Grant-Reed scheme, this section also clearly revealed the
underlying suboptimal simplifications made en route to the final
result. In the following, we will compare it to the standard
discrete SISO MUD based on the theory of VFEM and the associated
scheduling rules.

\subsubsection{The Standard Forms of Discrete SISO
MUD}\label{sec:stand_discrete_siso}

\begin{table}
\caption{Three scheduling schemes of turbo MUD employing discrete
SISO MUD.}\label{tab:disc_siso} \small
\begin{center}
\setlength{\tabcolsep}{3pt}
\begin{tabular}{|lllllcl|}  \hline
\multicolumn{7}{|c|}{\textbf{Sequential-Discrete-SISO}} \\
\hline
\multicolumn{7}{|l|}{\emph{Initialization}: $\mv = \zerov$ and $\LLR_{dec}(b_k)=0$ for all $k$}\\
\multicolumn{7}{|l|}{FOR $j = 1:J$ (\emph{Outer Iteration})}\\
    &\multicolumn{6}{l|}{FOR $k=1:K$} \\
    &   &\multicolumn{5}{l|}{$\LLR_{dec}(b_k) = 0$} \\
    &   &\multicolumn{5}{l|}{FOR $i = 1:I$ (\emph{Inner Iteration})} \\
    &   & & \multicolumn{4}{l|}{FOR $l=k:K$, $1:k-1$} \\
    &   & & & \multicolumn{3}{l|}{$\LLR_{mud}(b_l) = \frac{2}{\sigma^2}\left[\etav_k^T \rv - \betav_k^T\mv\right]$} \\
    &   & & & \multicolumn{3}{l|}{$\LLR_{pos}(b_l) = \LLR_{dec}(b_l) + \LLR_{mud}(b_l)$} \\
    &   & & & \multicolumn{3}{l|}{$m_l = \tanh[\LLR_{pos}(b_l)/2]$} \\
    &   & &END&&&\\
    &   &END&&&&\\
    &   & \multicolumn{5}{l|}{$\LLR_{dec}(b_k) \stackrel{\scriptsize \mbox{Decoding}}{\Longleftarrow} \LLR_{mud}(b_k)$} \\
    &END&&&&&\\
END&&&&&&\\ \hline \hline

\multicolumn{7}{|c|}{\textbf{Flooding-Discrete-SISO}} \\
\hline
\multicolumn{7}{|l|}{\emph{Initialization}: $\mv = \zerov$ and $\LLR_{dec}(b_k)=0$ for all $k$}\\
\multicolumn{7}{|l|}{FOR $j = 1:J$ (\emph{Outer Iteration})}\\
    &\multicolumn{6}{l|}{FOR $i = 1:I$ (\emph{Inner Iteration})} \\
    &   &\multicolumn{5}{l|}{FOR $k=1:K$} \\
    &   & & \multicolumn{4}{l|}{$\LLR_{pos}(b_k) = \LLR_{dec}(b_k) + \frac{2}{\sigma^2}\left[\etav_k^T \rv - \betav_k^T\mv\right]$} \\
    &   & & \multicolumn{4}{l|}{$m_k = \tanh[\LLR_{pos}(b_k)/2]$} \\
    &   &END&&&&\\
    &END&&&&&\\
    &\multicolumn{6}{l|}{FOR $k = 1:K$} \\
    &   & \multicolumn{5}{l|}{$\LLR_{mud}(b_k) = \LLR_{pos}(b_k) - \LLR_{dec}(b_k)$} \\
    &   & \multicolumn{5}{l|}{$\LLR_{dec}(b_k) \stackrel{\scriptsize \mbox{Decoding}}{\Longleftarrow} \LLR_{mud}(b_k)$} \\
    &END&&&&&\\
END&&&&&&\\
\hline\hline

\multicolumn{7}{|c|}{\textbf{Hybrid-Discrete-SISO}} \\
\hline
\multicolumn{7}{|l|}{\emph{Initialization}: $\mv = \zerov$ and $\LLR_{dec}(b_k)=0$ for all $k$}\\
\multicolumn{7}{|l|}{FOR $j = 1:J$ (\emph{Outer Iteration})}\\
    &\multicolumn{6}{l|}{FOR $k=1:K$} \\
    &   &\multicolumn{5}{l|}{$\LLR_{dec}(b_k) = 0$} \\
    &   &\multicolumn{5}{l|}{FOR $i = 1:I$ (\emph{Inner Iteration})} \\
    &   & & \multicolumn{4}{l|}{FOR $l=k:K$, $1:k-1$} \\
    &   & & & \multicolumn{3}{l|}{$\LLR_{mud}(b_l) = \frac{2}{\sigma^2}\left[\etav_k^T \rv - \betav_k^T\mv\right]$} \\
    &   & & & \multicolumn{3}{l|}{$\LLR_{pos}(b_l) = \LLR_{dec}(b_l) + \LLR_{mud}(b_l)$} \\
    &   & & & \multicolumn{3}{l|}{$m_l = \tanh[\LLR_{pos}(b_l)/2]$} \\
    &   & &END&&&\\
    &   &END&&&&\\
    &END&&&&&\\
    &\multicolumn{6}{l|}{FOR $k = 1:K$} \\
    &   & \multicolumn{5}{l|}{$\LLR_{dec}(b_k) \stackrel{\scriptsize \mbox{Decoding}}{\Longleftarrow} \LLR_{mud}(b_k)$} \\
    &END&&&&&\\
END&&&&&&\\ \hline

\end{tabular}
\setlength{\tabcolsep}{5pt}
\end{center}
\normalsize
\end{table}

In Table \ref{tab:disc_siso}, we summarize three different versions
of the standard discrete SISO MUD. The following highlights the
major characteristics of each scheme.

\emph{Sequential-Discrete-SISO}: The sequential schedule obtains the
EXT for $b_k$ through a serial update algorithm governed by
(\ref{eq:disc_LLR}). Before the inner iterations,
$\{\LLR_{dec}(b_l)\}_{l\neq k}$ are set to the most recent output
from the APP decoder, except for $\LLR_{dec}(b_k)$, which is set to
$0$. This is equivalent to setting $\xi_k = 1/2$ in
(\ref{eq:mean_field_postulate}), as required by the sequential
scheduling rule.

After the serial update, $\LLR_{mud}(b_k)$ is immediately sent to
the $k$-th APP decoder ($\{\LLR_{mud}(b_l)\}_{l\neq k}$ are
discarded), such that an updated prior $\LLR_{dec}(b_k)$ is
generated (see \fig \ref{fig:sp_view}). The sequential schedule is
inefficient, since a different serial update of $\LLR_{mud}(b_k)$
needs to be done $K$ times, one for each user. A SIC-based turbo MUD
scheme proposed by Kobayashi, Boutros, and Caire \cite{Koba01} can
be seen as a simplification to the full-blown
sequential-discrete-SISO MUD, with $I=1$ inner iteration.

\emph{Flooding-Discrete-SISO}: The flooding schedule is much more
efficient. The serial update algorithm in the inner iteration
updates the posterior LLR's, $\{\LLR_{pos}(b_k)\}_{k=1}^K$. After
$I$ iterations, in which the free energy is monotonically reduced,
reliable estimates of $\{\LLR_{pos}(b_k)\}_{k=1}^K$ are attained.
The SISO detector passes the EXT, $\LLR_{mud}(b_k)$, into the APP
decoder, where the decoding of $K$ users can be done in parallel.

\emph{Hybrid-Discrete-SISO}: The hybrid schedule differs from the
sequential schedule, in that when $\LLR_{mud}(b_k)$ is found, it is
not immediately sent to the APP decoder to update $\LLR_{dec}(b_k)$,
but stored until all other users' EXT's are obtained, to facilitate
parallel APP decoding.

%In Table \ref{tab:disc_siso}, we summarize both the
%sequential-schedule and flooding-schedule Discrete SISO MUD to
%highlight the underlying distinctions between the two. The
%hybrid-schedule approach is a simple modification to the
%sequential-schedule approach, and will be omitted.

\subsection{VFEM Interpretation of Decorrelating-Decision-Feedback SISO Multiuser
Detector}\label{sec:siso_ddf}

In \cite{Duel93}, Duel-Hallen proposed the
decorrelating-decision-feedback (DDF) multiuser detector. It has
been shown to out-perform most of the linear and
interference-cancellation detectors, especially in terms of near-far
resistance. However, the soft-decision DDF MUD and its application
within the turbo MUD framework is relatively unknown. In this
section, we will propose a SISO DDF multiuser detector using the
VFEM principle. The subsequent discussion will allow new insights
and new algorithms, including an interesting link to the discrete
SISO detector discussed earlier.

Consider applying the VFEM routine to the following postulated
distributions:
    \begin{equation}\label{eq:ddf_postulate}
    \left\{ \begin{array}{rcl}
    p(\bv) &=& \prod_{k=1}^K
    \xi_k^{\frac{1+b_k}{2}}(1-\xi_k)^{\frac{1-b_k}{2}}, \ \ b_k
    \in \{\pm 1\}\\
    p(\bar{\yv}|\bv) &=& \mathcal{N}(\Fv\Av\bv,\sigma^2 \Iv)\\
    Q(\bv) &=& \prod_{k=1}^K
    \gamma_k^{\frac{1+b_k}{2}}(1-\gamma_k)^{\frac{1-b_k}{2}}, \ \ b_k
    \in \{\pm 1\},
    \end{array} \right.
    \end{equation}
Notice that these distributions are identical to the discrete SISO
case in (\ref{eq:mean_field_postulate}), except that the received
vector $\rv$ is replaced by its sufficient statistics $\bar{\yv}$
(defined in (\ref{eq:DF_CDMA})). Therefore, we may directly make use
of the derivation in Section \ref{sec:mean_field}, and arrive at an
iterative detector similar to (\ref{eq:mean_gamma}):
\BEQ\label{eq:mean_gamma_ddf}
    \log\frac{1+m_k}{1-m_k} = \log\frac{1+\tilde{b}_k}{1-\tilde{b}_k} +
    \frac{2}{\sigma^2}\left[\bar{\etav}_k^T \bar{\yv} -
    \bar{\betav}_k^T\mv\right],
\EEQ where $\bar{\etav}_k$ and $\bar{\betav}_k$ are the $k$-th
column vector of $\Fv\Av$ and $\Av^T\Fv^T\Fv\Av -
\diag{(\Av^T\Fv^T\Fv\Av)}$, respectively. (\ref{eq:mean_gamma_ddf})
is in fact identical to (\ref{eq:mean_gamma}), since $\Fv^T\bar{\yv}
= \Sv^T\rv = \yv$ and $\Fv^T\Fv = \Sv^T\Sv=\Rv$.

Consider the uncoded scenario, \ie, $\tilde{b}_k = 0$ for
$k=1,\cdots,K$, then (\ref{eq:mean_gamma_ddf}) reduces to
\BEQ\label{eq:hypertan_ddf}
    m_k = \tanh \left\{\frac{1}{\sigma^2}\left(\bar{\etav}_k^T \bar{\yv} -
    \bar{\betav}_k^T\mv\right)\right\}.
\EEQ
The free energy is monotonically reduced if $m_k$ is evaluated in a
SIC fashion similar to (\ref{eq:LLRupdate}), \ie, in the $i$-th
iteration: \BEQ\label{eq:mk_update_DDF1}
    m_k^{(i)} = \tanh \left\{\frac{1}{\sigma^2}\left(\bar{\etav}_k^T \bar{\yv} -
    \bar{\betav}_k^T\mv_{<k}^{(i)} - \bar{\betav}_k^T\mv_{>k}^{(i-1)}\right)\right\}.
\EEQ

Now we take a crucial step that will produce the DDF SISO detector
based on (\ref{eq:mk_update_DDF1}). We will alter the definition of
$\bar{\etav}_k$ and $\bar{\betav}_k$, by replacing $\Fv$ with a new
matrix $\Fv_k$. Let $\Fv_k$ be $\Fv$, except with elements
$F_{k+1,k}$ to $F_{K,k}$ nulled, \ie, \BEQ\label{eq:F_k}
    \Fv_k = \left[\BA{cccccc}F_{1,1}& & & & &\\
    F_{2,1}&\ddots& & & &\\
    & & F_{k,k} & & &\\
    \vdots & & \fbox{0} & F_{k+1,k+1} & &\\
    & & \vdots & \vdots & \ddots &\\
    F_{K,1}& & \fbox{0} & F_{K,k+1} & & F_{K,K}\EA\right].
\EEQ Then we let $\bar{\etav}_k$ and $\bar{\betav}_k$ be the $k$-th
column vectors of $\Fv_k\Av$ and $\Av^T\Fv_k^T\Fv_k\Av -
\diag{(\Av^T\Fv_k^T\Fv_k\Av)}$, respectively. Subsequently, we see
that \BEA
    \bar{\etav}_k &=&  [0,\cdots,0,A_k F_{k,k},0,\cdots,0]^T\\
    \bar{\betav}_k &=& A_k F_{k,k}[A_1 F_{k,1},\cdots,A_{k-1}
    F_{k,k-1},0,\cdots,0]^T,
\EEA
and \BEQ\BA{ccc}
    \bar{\etav}_k^T \bar{\yv} &=& A_k F_{k,k} \bar{y}_k\\
    \bar{\betav}_k^T\mv_{>k} &=& 0.
\EA\EEQ %
Hence (\ref{eq:mk_update_DDF1}) becomes
\BEQ\label{eq:mk_update_DDF2}
    m_k = \tanh\left\{\frac{1}{\sigma^2}\left(A_k F_{k,k} \bar{y}_k -
    \bar{\betav}_k^T\mv_{<k}\right)\right\}. \EEQ
Finally, it is not difficult to recognize that the argument inside
the $\tanh(\cdot)$ function is the DDF detector output, scaled by
$1/\sigma^2$. The additional $\tanh$ function has its intuitive
appeal, since it provides soft bit estimates for BPSK in an AWGN
channel (assuming perfect cancellation of interference).

We have therefore derived a soft-decision DDF detector, obtaining it
based on the discrete SISO MUD, and by replacing the matrix $\Fv$
with $\Fv_k$ in the free energy minimization stage. If informative
priors, $\{\LLR^{(0)}(b_k)\}_{k=1}^K$, are used,
(\ref{eq:mk_update_DDF2}) simply becomes
\BEQ\label{eq:mk_update_DDF3}
    \LLR_{pos}(b_k) = \LLR^{(0)}(b_k) + \frac{2}{\sigma^2}\left(A_k F_{k,k} \bar{y}_k -
    \bar{\betav}_k^T\mv_{<k}\right),
\EEQ
where $m_k = \tanh[\LLR_{pos}(b_k)/2]$. But the EXT is unchanged,
since
\BEQ\label{eq:mk_update_DDF4}
    \LLR_{mud}(b_k) = \LLR_{pos}(b_k) - \LLR^{(0)}(b_k) = \frac{2}{\sigma^2}\left(A_k F_{k,k} \bar{y}_k -
    \bar{\betav}_k^T\mv_{<k}\right).
\EEQ

\begin{boxitjournal}
\begin{definition}
A \emph{DDF SISO Multiuser Detector} is a multiuser detector that
obtains soft estimates $Q(\bv)$ through the VFEM routine, subject to
the following postulated distributions:
    \begin{equation}\label{eq:postulate_ddf}
    \left\{ \begin{array}{rcl}
    p(\bv) &=& \prod_{k=1}^K
    \xi_k^{\frac{1+b_k}{2}}(1-\xi_k)^{\frac{1-b_k}{2}}, \ \ b_k
    \in \{\pm 1\}\\
    p(\bar{\yv}|\bv) &=& \mathcal{N}(\Fv\Av\bv,\sigma^2 \Iv)\\
    Q(\bv) &=& \prod_{k=1}^K
    \gamma_k^{\frac{1+b_k}{2}}(1-\gamma_k)^{\frac{1-b_k}{2}}, \ \ b_k
    \in \{\pm 1\},
    \end{array} \right.
    \end{equation}
and replacing $\Fv$ with $\Fv_k$ (as in (\ref{eq:F_k})) in the free
energy minimization stage.
\end{definition}
\end{boxitjournal}

To understand the effect of transforming $\Fv$ to $\Fv_k$ in simple
terms, we first need to realize that $\bar{\etav}_k$ acts as a
matched filter on $\bar{\yv}$ to extract the decision metric for
$b_k$, while $\bar{\betav}_k$ indicates the amount of interference
to be subtracted from $\bar{\etav}_k^T \bar{\yv}$ to improve
estimation. By defining $\Fv_k$ as in (\ref{eq:F_k}), we
heuristically assume $\bar{\etav}_k$ to be non-zero only in the
$k$th element, essentially ignoring the presence of $b_k$ in
$\bar{y}_{k+1},\cdots,\bar{y}_K$. This simplification also makes
$\bar{\betav}_k$ non-zero only in the first $k-1$ elements, implying
that only the estimates for $b_1,\cdots,b_{k-1}$ need to be
subtracted for the detection of $b_k$.

In this section, we have shown that the same free energy expression
specifies the DDF SISO MUD and discrete SISO MUD. But in DDF SISO
MUD, we replaced $\Fv$ with $\Fv_k$ in the execution of the
coordinate descent algorithm to enable a decision-feedback
structure.

\subsection{DDF-Aided Discrete SISO MUD}\label{sec:modify_disc}

%Now that it is clear that DDF SISO MUD is an approximate
%implementation of discrete SISO MUD, then when should we ever
%choose the suboptimal version (DDF SISO) over the exact version
%(discrete SISO)? The answer to this question is most easily seen
%from the point of view of free energy minimization. The free
%energy expression in (\ref{eq:f_meanfield}) is a non-convex
%function in terms of $\mv$. While both discrete SISO MUD and DDF
%SISO MUD may be interpreted as coordinate descent algorithms to
%reduce (\ref{eq:f_meanfield}), by replacing $\Fv$ with $\Fv_k$ in
%the process, DDF SISO MUD is less likely to be trapped in poor
%local minima, and therefore often outperforms discrete SISO MUD.
%
%Then is it possible to combine the good convergence property of the
%DDF SISO MUD and the optimality promised by the discrete SISO MUD?
%The answer is yes. Although DDF SISO MUD converges after the first
%iteration, it has not reached the minimum of (\ref{eq:f_meanfield}).
%Therefore, with the help of good initialization with DDF SISO MUD,
%it is possible to implement the discrete SISO MUD in subsequent
%iterations to drive the free energy even lower, and produce improved
%performance. In this section, we will introduce a so-called
%\emph{DDF-aided discrete SISO} detector, which is obtained by
%replacing the first iteration of the original discrete SISO MUD
%algorithm by DDF SISO MUD.

Having discussed both Gaussian SISO MUD and discrete SISO MUD, a
natural question to ask is how the two compare with each other in
complexity and performance. In general, it is well-known that
Gaussian SISO MUD is a robust algorithm, but has relatively high
complexity. This is easy to see from the VFEM viewpoint, since
Gaussian SISO MUD minimizes $\mathcal{F}_{gauss}(\muv,\Sigmav)$
exactly in the free energy minimization stage. But solving the
optimization problem exactly entails higher complexity due to the
need for matrix inversion.

Discrete SISO MUD, on the other hand, decreases
$\mathcal{F}_{disc}(\mv)$ iteratively through a SIC-like procedure,
which has only linear complexity in $K$. However, since
$\mathcal{F}_{disc}(\mv)$ is not a convex function, serial updates
of this form can become trapped in local minima, resulting in poor
detection performance. This is the reason why detection schemes
based on discrete SISO MUD only work in systems with small spreading
code correlations, such as random spreading CDMA systems
\cite{Alex98,Koba01}, but not in high-interference channels, such as
those considered in \cite{Wangx99} with the spreading code
correlation between different users being $\rho=0.7$.

As is shown in Section \ref{sec:siso_ddf}, the DDF SISO MUD has the
same target free energy as the discrete SISO MUD. But through a
small modification to the coordinate descent step, the DDF SISO MUD
is able to overcome local minima and bring the free energy close to
it's minima. The only problem is that, due to the substitution of
$\Fv_k$ for $\Fv$, it cannot refine its estimate iteratively, and
hence does not settle at the exact global minimum of
$\mathcal{F}_{disc}(\mv)$.

To combine the good convergence property of the DDF SISO MUD and the
optimality promised by the discrete SISO MUD, we will introduce a
so-called \emph{DDF-aided discrete SISO} detector as follows
\newline

\begin{boxitjournal}
\begin{definition}
A \emph{DDF-Aided Discrete SISO Multiuser Detector} is a multiuser
detector that obtains soft estimates $Q(\bv)$ through the VFEM
routine, subject to the postulated distributions in
(\ref{eq:postulate_ddf}). It is implemented by replacing the first
iteration of the discrete SISO MUD algorithm with DDF SISO MUD.
\end{definition}
\end{boxitjournal}

This detector utilizes the good initialization of DDF SISO MUD, and
implements discrete SISO MUD in subsequent iterations to drive the
free energy even lower, and produce improved performance. We will
demonstrate how the DDF-aided discrete SISO MUD improves upon both
the original discrete SISO MUD and DDF SISO MUD through a simple
simulation based on Example 1 in \cite{Duel93}.

\begin{figure}

\hspace{-0.5in}\scalebox{1.0}{\includegraphics{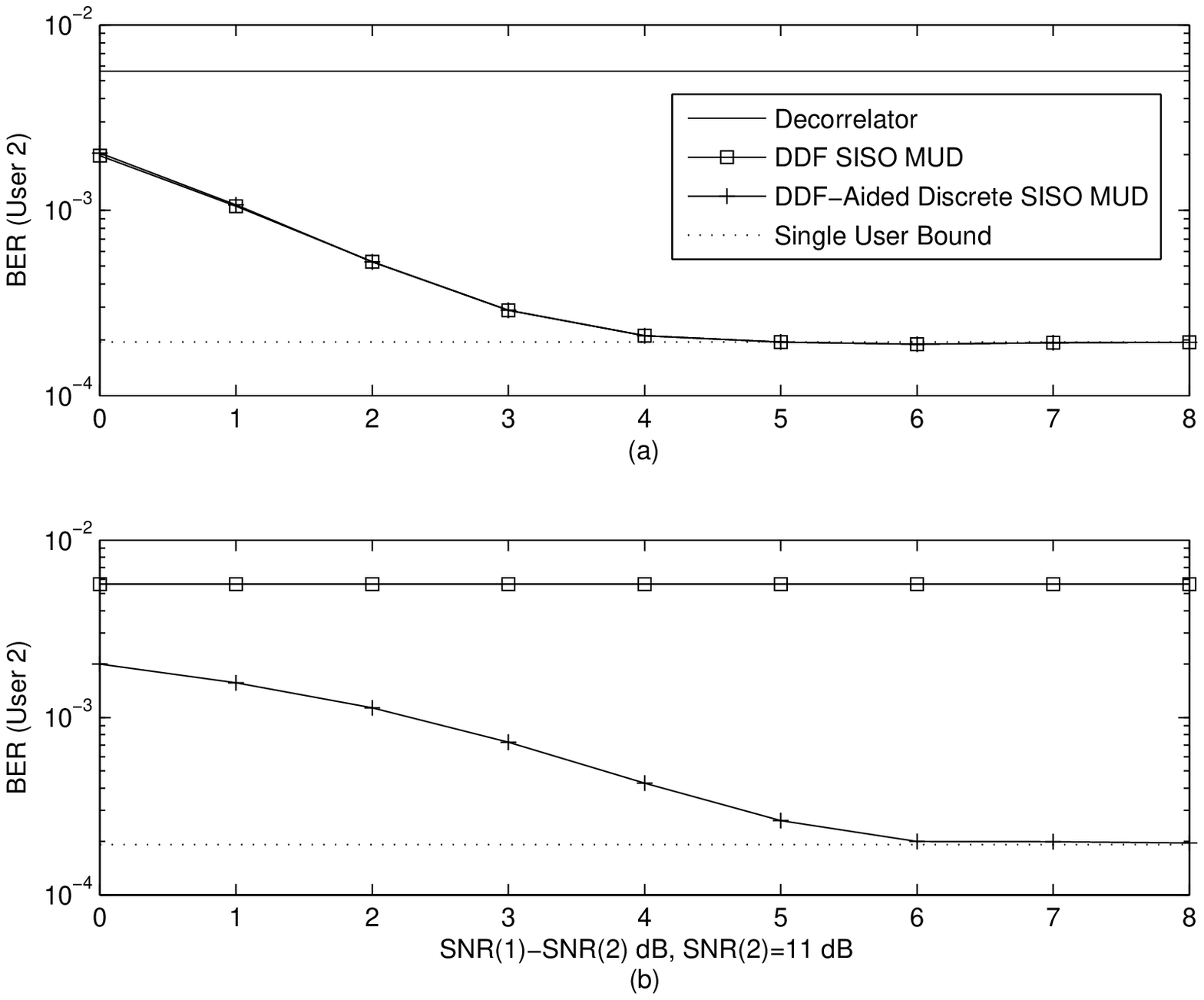}}\caption{BER
vs. SNR difference in two-user channel. (a) Strong user detected
first. (b) Weak user detected first.}\label{fig:uncode_ddf}

\end{figure}

Consider an uncoded two-user synchronous CDMA system with
spreading-sequence cross-correlation $\rho=0.7$. We let the
signal-to-noise ratio (SNR) of user 1 (strong user) be no smaller
than that of user 2 (weak user), \ie, $\SNR(1) \geq \SNR(2)$. Fixing
$\SNR(2)$ at $11$ dB and varying $\SNR(1)$, we obtain the bit error
rate (BER) of user 2 as shown in \fig \ref{fig:uncode_ddf}. In \fig
\ref{fig:uncode_ddf}(a), we detect the strong user first. the
DDF-aided discrete SISO MUD is implemented with the DDF SISO MUD
followed by four discrete SISO MUD iterations. It is seen that in
this case, the additional discrete SISO MUD iterations do not offer
performance enhancement, thus the DDF-aided discrete SISO MUD
performs nearly identical to DDF SISO MUD. The use of discrete SISO
MUD alone suffers from poor convergence as predicted, and is omitted
in the plot.

However, like Duel-Hallen's conventional DDF detector, DDF SISO
detector is sensitive to the detection order. This means if the weak
user is detected first, the near-far resistance property no longer
holds. Such an effect is depicted in \fig \ref{fig:uncode_ddf}(b),
where we detect the weak user first. By fixing $\SNR(2)$ at $11$ dB
and varying $\SNR(1)$, the BER of different schemes are plotted. It
is seen that the DDF SISO MUD no longer approaches near-optimal
performance as the SNR difference increases, while the DDF-aided
discrete SISO MUD continues to demonstrate good near-far resistance
even with non-optimal detection order. This is because the
additional discrete SISO MUD iterations rectifies the performance
degradation of DDF SISO MUD due to unfavourable detection order.

These simple examples reveal that the DDF-aided discrete SISO MUD is
a much more powerful detection scheme compared to both discrete SISO
MUD and DDF SISO MUD. It can be viewed as either the
multiple-iteration extension to DDF SISO MUD, or as a discrete SISO
MUD with convergence acceleration. The DDF-aided discrete SISO MUD
is a powerful algorithm that is now capable of coping with
strong-interference channels such as the ones assumed in
\cite{Wangx99}, which will be studied in Section
\ref{sec:simulation1}.

\subsection{Summary}

\begin{table}[h]\caption{Variational-inference-based multiuser detectors employing different scheduling schemes.} \label{tab:var_comparison}
\begin{center}
\begin{tabular}{|c|c|c|} \hline
    & \bf{Gaussian SISO} & \bf{Discrete SISO}  \\
    \hline\hline
    Sequential & $\star$ & \cite{Koba01} \\
    \hline
    Flooding & $\star$ & \cite{Alex98}  \\
    \hline
    Hybrid & \cite{Wangx99} & $\star$ \\
\hline
\end{tabular}
\end{center}
\end{table}

Table \ref{tab:var_comparison} categorizes some of the existing
turbo multiuser detectors according to the standard-form SISO MUD
schemes outlined in Table \ref{tab:disc_siso} and
\ref{tab:gauss_siso}. $\star$ indicates the schemes that are
outcomes of the general framework, but not seen in the literature.
We now outline how the existing schemes fit into the categories
created.

\BIT

\item \cite{Wangx99} is identical to the hybrid-Gaussian -SISO
MUD, but it is re-derived in Section \ref{sec:exist_gaussian_siso}
via a completely different VFEM-based approach. With the help of
the insights offered by the VFEM framework, we are able to further
extend \cite{Wangx99} to sequential and flooding schedule
implementations, both explained in Section
\ref{sec:stand_gaussian_siso}.

\item \cite{Koba01} can be seen as the standard sequential-discrete-SISO MUD with
sequential scheduling, and $I=1$ inner iteration. Moreover,
\cite{Koba01} considers the joint estimation of noise-plus-MAI
variance and channel amplitude. Like the noise-plus-MAI variance
estimation in \cite{Alex98}, this can also be studied within the
VFEM framework, as an instance of the variational EM algorithm
discussed in Section \ref{sec:var_EM}.

\item \cite{Alex98} can be seen as a simplified version of
flooding-discrete-SISO MUD, but it differs from the the standard
approach in two aspects: 1) \cite{Alex98} uses parallel updates of
each user's bit LLR, which does not guarantee the reduction of
free energy. 2) \cite{Alex98} uses the posterior estimate (instead
of the extrinsic information) from the APP decoder as the
initialization of $\mv$, a practice that is suboptimal from the
message-passing standpoint. \EIT

%%%%%%%%%%%%%%%%%%%%%%%%%%%%%%%%%%%%%%%%%%%%%%%%%%%%%%%%%%%%%%%%%%
\section{Variational EM for Iterative Parameter
Estimation}\label{sec:var_EM}

In recent years, the impact of imperfect channel estimation on
uncoded and coded multiuser detection have been studied in the
large system limit \cite{Evans00,Li04,Li05}. However, the
alleviation of this problem has rarely been systematically
investigated in the literature. In this section, we will introduce
an important extension of variational inference, called
variational EM, to enable joint parameter estimation in turbo MUD.
The two examples in Sections \ref{sec:em_sisommse} and
\ref{sec:em_meanfield}, based on the Gaussian SISO detector and
discrete SISO detector, respectively, will illustrate how the
variational EM framework provides a feasible solution to practical
turbo receiver design when exact channel state information (CSI)
is unavailable.

\subsection{Formulation of Variational EM Algorithm}

Like most detectors, variational-inference-based detection schemes
assume perfect knowledge of system parameters, such as various
types of channel information. These parameters, in practice
however, may not be known accurately at the receiver. One
traditional way to incorporate the uncertainties of these
parameters in the detection operation is through the EM algorithm
\cite{Demp77,Kocian03}. The EM algorithm is used to estimate a
vector of parameters, say $\thetav$, from the observation $\rv$
that is termed ``incomplete data'', together with some auxiliary
or hidden variable, say $\bv$. The algorithm iteratively carries
out two operations: the E step and the M step. The $j$-th
iteration effectively computes a probability density
$p(\bv|\rv,\thetav^{(j-1)})$ in the E step, where
$\thetav^{(j-1)}$ is the estimate of $\thetav$ in the previous
iteration, and then in the M step maximizes
\begin{equation}
  U(\thetav,\thetav^{(j-1)}) = \int p(\bv|\rv,\thetav^{(j-1)})
  \log p(\bv,\rv|\thetav) d\bv
\end{equation}
over $\thetav$, yielding $\thetav^{(j)}$.

The work in \cite{Neal98} shows that the
EM algorithm is equivalent to jointly estimating the hidden
variables and parameters by minimizing a single free-energy
expression over a postulated distribution for the hidden
variables, and over the parameters. The VFEM formulation offers an
additional degree of freedom to the conventional EM algorithm,
such that in the E step, an approximate posterior
$Q(\bv|\thetav^{(j-1)})$ may be used to replace the exact
posterior $p(\bv|\rv,\thetav^{(j-1)})$. Variational EM has been
successfully applied in various applications, \eg, in image
processing to perform scene analysis \cite{Frey05}, and in joint
detection/estimation problems in wireless channels
\cite{Lin05detect,Niss06}.

To provide a concrete example, assume $\thetav$ remains static
over $T$ independent uses of the channel. In the context of MUD,
this implies that we assume $\thetav$, the noise variance
$\sigma^2$ for example, remains constant when a block of $T$ bits
are transmitted by each user ($T$ could be the code word length).
The variational EM algorithm extracts point estimates for
$\thetav$ and postulated posterior distributions over the channel
bits. Therefore, the new $Q$-function may be written as:
\begin{equation}\label{eq:var_em}
    Q(\bv_{1},\cdots,\bv_{T},\thetav) = \delta(\thetav-\hat{\thetav})\prod_{t=1}^T
    Q(\bv_t),
\end{equation}
where $\hat{\thetav}$ is an estimate of $\thetav$, and $\bv_t$
contains the channel bits transmitted in the $t$-th use of the
channel. The notation
$\delta(\av-\hat{\av})$ denotes a vector Dirac delta function with
the following properties: $\int\delta(\av-\hat{\av})f(\av)\ d\av =
f(\hat{\av})$, and $\int\delta(\av-\hat{\av})\ d\av = 1$. Recall that for i.i.d. data, $p(\bv,\thetav,\rv) =
p(\thetav)\prod_{t=1}^T p(\bv_t,\rv_t|\thetav)$. Substituting
(\ref{eq:var_em}) into (\ref{eq:free_F}) yields the following free
energy:
\begin{equation}\label{eq:var_free}\BA{rcl}
     \mathcal{F} &=& \displaystyle \int_{\bv}\int_{\thetav}
    Q(\bv_{1},\cdots,\bv_{T},\thetav) \log \frac{Q(\bv_{1},\cdots,\bv_{T},\thetav)}{p(\bv_1,\cdots,\bv_T,\thetav,\rv)}
    d\thetav d\bv\\
    &=& \displaystyle \int_{\bv}\int_{\thetav}
    \delta(\thetav-\hat{\thetav})\prod_{t=1}^T
    Q(\bv_t) \log \frac{\delta(\thetav-\hat{\thetav})\prod_{t=1}^T
    Q(\bv_t)}{p(\thetav)\prod_{t=1}^T p(\bv_t,\rv_t|\thetav)}
    d\thetav d\bv\\
    &=& \displaystyle \int_{\thetav} \delta(\thetav-\hat{\thetav}) \log \delta(\thetav-\hat{\thetav})
    d\thetav - \int_{\thetav} \delta(\thetav-\hat{\thetav}) \log p(\thetav)
    d\thetav + \int_{\bv} \prod_{t=1}^T Q(\bv_t) \log \frac{\prod_{t=1}^T
    Q(\bv_t)}{\prod_{t=1}^T p(\bv_t,\rv_t|\hat{\thetav})} d\bv\\
    &=& \displaystyle -\log p(\hat{\thetav}) + \sum_{t=1}^T\left( \int_{\bv_t} Q(\bv_t)\log
    \frac{Q(\bv_t)}{p(\bv_t,\rv_t|\hat{\thetav})} d\bv_t
    \right).
\EA\end{equation}

In the last line of the above equation we omit the constant term
$\int_{\thetav} \delta(\thetav-\hat{\thetav}) \log
\delta(\thetav-\hat{\thetav}) d\thetav$. The term
$p(\hat{\thetav})$ constitutes the prior knowledge of the
parameter. In cases when it is not available, we may set
$p(\hat{\thetav})=constant$ and ignore it in the minimization of
free energy.

As proven in \cite{Neal98}, alternating between minimizing
(\ref{eq:var_free}) {\wrt} $\{Q(\bv_t)\}_{t=1}^T$ in the E step, and
{\wrt} $\hat{\thetav}$ in the M Step leads to the exact EM
algorithm where $\{\bv_t\}_{t=1}^T$ are the ``hidden variables''
and $\hat{\thetav}$ is the unknown parameter of interest.
Unfortunately, the exact EM is only possible in special cases,
because the computation in the E step of $Q(\bv_t) =
p(\bv_t|\rv_t,\hat{\thetav})$ (s.t. $\int_{\bv_t} Q(\bv_t) d\bv_t
=1$) is often intractable. But suppose we use a postulated (and
simple) distribution $Q(\bv_t)$, with parameter $\lambda_t$, and
then find $\lambda_t$ that minimizes (\ref{eq:var_free}). We then
arrive at the \emph{variational EM} algorithm, which consists of
the initialization plus the E step and M step in the $j$-th
iteration:

\begin{boxitjournal}
\textbf{Initialization} Choose initial values for
$\hat{\thetav}^{(0)}$.

\textbf{E Step} Minimize
$\mathcal{F}(\lambda_1,\cdots,\lambda_T,\hat{\thetav}^{(j-1)})$ in
(\ref{eq:var_free}) {\wrt} $\lambda_t$
\begin{equation}\label{eq:estep}
    \lambda_t^{(j)} = \arg \min_{\lambda_t} \int_{\bv_t} Q(\bv_t)\log
    \frac{Q(\bv_t)}{p(\bv_t,\rv_t|\hat{\thetav}^{(j-1)})} d\bv_t,
\end{equation}
for $t = 1,\cdots,T$.

\textbf{M Step} Minimize
$\mathcal{F}(\lambda_1^{(j)},\cdots,\lambda_T^{(j)},\hat{\thetav})$
in (\ref{eq:var_free}) {\wrt} $\hat{\thetav}$
\begin{equation}\label{eq:mstep}
    \hat{\thetav}^{(j)} = \arg \min_{\hat{\thetav}} \sum_{t=1}^T\left( \int_{\bv_t} Q^{(j)}(\bv_t)\log
    \frac{Q^{(j)}(\bv_t)}{p(\bv_t,\rv_t|\hat{\thetav})} d\bv_t
    \right)-\log p(\hat{\thetav}).
\end{equation}
\end{boxitjournal}

In the rest of the section, we will implement the variational EM
algorithm in both Gaussian SISO MUD and discrete SISO MUD, to
resolve the uncertainty in channel information at the receiver.
More specifically, we will assume no noise variance information
and noisy channel amplitude information at the receiver, and
attempt to adaptively estimate the noise variance, as well as
improve the channel amplitude estimation, jointly with data
detection.

\subsection{Channel and Noise Variance Estimation for Gaussian SISO
MUD}\label{sec:em_sisommse}

Adding a time index $t$ to (\ref{eq:CDMA}) to represent a sequence
of channel observations $\rv_t = \Sv\Av\bv_t + \nv_t$
($t=1,\cdots,T$), according to (\ref{eq:var_free}), we may write
the free energy for $T$ channel realizations as
\BEQ\label{eq:F_gauss}
    \mathcal{F}_{gauss}(\muv_1,\cdots,\muv_T,\Sigmav_1,\cdots,\Sigmav_T,\sigma^2,\av) =
     - \log p(\sigma^2) - \log p(\av) + \sum_{t=1}^T \mathcal{F}_{gauss}(\muv_t,\Sigmav_t|\sigma^2,\av).
\EEQ
where $\av = \diag(\Av)$, and we define \BEQ
    \mathcal{F}_{gauss}(\muv_t,\Sigmav_t|\sigma^2,\av) = \int_{\bv_t} Q(\bv_t)\log
\frac{Q(\bv_t)}{p(\bv_t,\rv_t|\sigma^2,\av)} d\bv_t , \EEQ
which is equal to (\ref{eq:F_siso}), except that $\sigma^2$ and
$\av$ are now explicitly shown to be variables of $\mathcal{F}$.
Here $\thetav = \{\sigma^2,\av\}$ are the model parameters to be
estimated. In (\ref{eq:F_gauss}), $p(\sigma^2)$ is a constant,
since we do not assume prior knowledge about $\sigma^2$. But
estimates of the channel, however noisy, can be assumed available
at the receiver. In particular, we model the true channel, $\av$,
as the sum of the channel estimate, $\tilde{\av}$, and random
measurement error with variance $\varsigma^2$ \cite{Medard00}. Or,
equivalently, $p(\av) = \mathcal{N}(\tilde{\av},\varsigma^2 \Iv)$, where
$\varsigma$ is assumed to be known.

The E step, \ie, the estimation of $\muv$ and $\Sigmav$, has already
been completed in Section \ref{sec:siso_mmse}. The only challenge
now remaining is the M step. From (\ref{eq:mstep}), we see that we
are required to solve for \BEQ
    \{\hat{\sigma}^2, \hat{\av}\} = \arg\min_{\sigma^2,\av}\
    \mathcal{F}_{gauss}(\muv_1,\cdots,\muv_T,\Sigmav_1,\cdots,\Sigmav_T,\sigma^2,\av).
\EEQ
To this end, the following identity is needed:

\begin{lemma}\label{lemma:schur_trace1}
    \BEQ
    \tr[\diag(\xv)\cdot\Av \cdot\diag(\yv)\cdot\Bv] = \xv^T(\Av\circ \Bv^T)\yv
    \EEQ
    for square matrices $\Av$, $\Bv \in \mathbb{R}^{N \times N}$,
    and vectors $\xv$, $\yv \in \mathbb{R}^{N \times 1}$.
\end{lemma}
\BP
    Writing $\Av = [A_{ij}]$ and $\Bv = [B_{ij}]$, it is easily verified that both sides of the equation are equal to $\sum_{i,j}x_i A_{ij}
B_{ji} y_j$. \EP

Utilizing Lemma \ref{lemma:schur_trace1} and ignoring the terms
independent of $\sigma^2$ and $\av$, we have \BEQ\BA{rl}
    &\mathcal{F}_{gauss}(\muv_1,\cdots,\muv_T,\Sigmav_1,\cdots,\Sigmav_T,\sigma^2,\av)\\
    =& \sum_{t=1}^T\left\{ \frac{1}{2\sigma^2}(\rv_t - \Sv\Muv_t\av)^T(\rv_t - \Sv\Muv_t\av) + \frac{N}{2}\log(\sigma^2) + \frac{1}{2\sigma^2} \av^T[(\Sv^T\Sv)\circ
    \Sigmav_t]\av + \frac{1}{2\varsigma^2}(\av-\tilde{\av})^T(\av-\tilde{\av})
    \right\},
\EA\EEQ
where $\Muv_t = \diag(\muv_t)$. Equating $\partial
\mathcal{F}/\partial \av = \zerov$ produces \BEQ\label{eq:av_hat}
    \hat{\av} = \left\{\sum_{t=1}^T \Muv_t^{T}\Sv^T\Sv\Muv_t + (\Sv^T\Sv)\circ\Sigmav_t + \frac{\sigma^2}{\varsigma^2}\Iv\right\}^{-1} \left( \sum_{t=1}^T
    \Muv_t\Sv^T\rv_t + \frac{\sigma^2}{\varsigma^2}\tilde{\av}\right).
\EEQ Substituting $\av = \hat{\av}$ into $\mathcal{F}$ and solving
for
$\mathcal{F}(\muv_1,\cdots,\muv_T,\Sigmav_1,\cdots,\Sigmav_T,\sigma^2,\hat{\av})/\partial
\sigma^{-2} = 0$ gives \BEQ\label{eq:sigma_hat}
    \hat{\sigma}^2 = \frac{1}{NT} \left\{\sum_{t=1}^T\left[ (\rv_t-\Sv\hat{\Av}\muv_t)^T (\rv_t-\Sv\hat{\Av}\muv_t) + \hat{\av}^T[(\Sv^T\Sv)\circ\Sigmav_t]\hat{\av} \right]
    \right\},
\EEQ
where $\hat{\Av} = \diag(\hat{\av})$. Note that (\ref{eq:av_hat})
and (\ref{eq:sigma_hat}) decrease the free energy in a coordinate
descent manner, not necessarily minimizing it, due to the coupling
of $\av$ and $\sigma^2$ in the two equations. But this is
acceptable, since they will converge to the exact minimizers over
the EM iterations.

After incorporating iterative decoding, the variational EM algorithm
for turbo MUD employing flooding-Gaussian-SISO MUD is described in
Table \ref{tab:gauss_siso_em}.

\begin{table}\caption{Variational EM algorithm employing Gaussian SISO MUD.}\label{tab:gauss_siso_em}
\centering \vspace{3pt}
\setlength{\tabcolsep}{3pt} \small
\begin{tabular}{|c|l|} \hline

\emph{Initialization} & \BT{c} Set $\tilde{\bv}_t = \zerov$,
$\sigma^{2(0)} = 0$, and $\av^{(0)} = \tilde{\av}$. \vspace{5pt} \ET\\

%%%%%%% The E Step
\BT{c} \emph{Update for $Q(\bv)$}\\ \emph{E Step} \ET &

\BT{lllcl}
    FOR& \multicolumn{4}{l}{ $j = 1:J$ (\emph{Outer Iteration})}\\
    &\multicolumn{4}{l}{FOR $k=1:K$ and $t=1:T$} \\
    &   & $\tilde{\bv}_{t,k}$ &=& $[\tilde{b}_{t,1},\cdots,\tilde{b}_{t,k-1},0,\tilde{b}_{t,k+1},\cdots,\tilde{b}_{t,K}]^T$\\
    &   & $\Wv_t$ &=&
            $\diag([1-\tilde{b}_{t,1}^2,\cdots,1-\tilde{b}_{t,K}^2]^T)$\\
    &   & $\check{\mu}_{t,k}$ &=& $A_k\ev_k^T\left[\Av^T\Wv_t\Av + \sigma^2\Rv^{-1}\right]^{-1} \left[\Rv^{-1}\yv_t-\Av\tilde{\bv}_{t,k}\right]$\\
    &   & $\check{\alpha}_{t,k}$ &=& $(1-\tilde{b}_{t,k}^2)A_k^2 \left[(\Av^T\Wv_t\Av + \sigma^2 \Rv^{-1})^{-1} \right]_{k,k}$ \\
    &   & \multicolumn{3}{l}{$\LLR_{mud}(b_{t,k}) = \frac{2\check{\mu}_{t,k}}{1-\check{\alpha}_{t,k}}$} \\
    &END& & &\\
\ET\\

%%%%%%% The Decoding
\BT{c} \emph{BCJR Decoding} \ET & \hspace{18pt}
\BT{lllcl}
    &\multicolumn{4}{l}{FOR $k=1:K$ and $t=1:T$} \\
    &   & \multicolumn{3}{l}{$\LLR_{dec}(b_{t,k}) \stackrel{\scriptsize \mbox{Decoding}}{\Longleftarrow} \LLR_{mud}(b_{t,k})$} \hfill (\emph{Extrinsic Information})\\
    &   & $\tilde{b}_{t,k}$ &=& $\tanh[\LLR_{dec}(b_{t,k})/2]$\\
    &   & $\hat{b}_{t,k}$ &=& $\tanh[\LLR_{mud}(b_{t,k})/2 + \LLR_{dec}(b_{t,k})/2 ]$ \hfill (\emph{Posterior Estimate})\\
    &   & $\mu_{t,k}$ &$\leftarrow$& $\hat{b}_{t,k}$\\
    &   & \multicolumn{3}{l}{$[\Sigmav_t]_{k,k} \leftarrow 1-\mu_{t,k}^2$}\\
    &END& & &\\
\ET\\

%%%%%%% The M Step
\BT{c} \emph{Update for $\sigma^2$ and $\av$}\\ \emph{M Step} \ET &
\BT{lllcl}
    &   & $\av$ &$\leftarrow$& $\left\{\sum_{t=1}^T \Muv_t^T\Sv^T\Sv\Muv_t^T + (\Sv^T\Sv)\circ\Sigmav_t  + \frac{\sigma^2}{\varsigma^2}\Iv\right\}^{-1} \left( \sum_{t=1}^T \Muv_t\Sv^T\rv_t
    + \frac{\sigma^2}{\varsigma^2}\tilde{\av} \right)$\\
    &   & $\sigma^2$ &$\leftarrow$& $\frac{1}{NT} \left\{\sum_{t=1}^T\left[ (\rv_t-\Sv\Av\muv_t)^T (\rv_t-\Sv\Av\muv_t) + \av^T[(\Sv^T\Sv)\circ\Sigmav_t]\av \right]
    \right\}$\\
\ET\\
& END\\
\hline
\end{tabular}
\setlength{\tabcolsep}{5pt} \normalsize
\end{table}

This is an extension to the flooding-Gaussian-SISO MUD algorithm
in Table \ref{tab:gauss_siso}. The variational EM extension to
sequential or hybrid Gaussian-SISO MUD is straightforward, where
the M step may be implemented either once every outer iteration
$j$, or more frequently, after the E step update of each user.

\subsection{Channel and Noise Variance Estimation for Discrete SISO
MUD}\label{sec:em_meanfield}

Similar to Gaussian SISO MUD, the free energy of discrete SISO MUD
for $T$ channel outputs can be written as \BEQ\label{eq:F_disc}
    \mathcal{F}_{disc}(\mv_1,\cdots,\mv_T,\sigma^2,\av) =
    - \log p(\sigma^2) - \log p(\av) + \sum_{t=1}^T \mathcal{F}_{disc}(\mv_t|\sigma^2,\av) ,
\EEQ

$\mathcal{F}_{disc}(\mv_t|\sigma^2,\av)$ is simply
(\ref{eq:f_meanfield}) with an additional time index $t$. In
(\ref{eq:F_disc}), we set $p(\sigma^2)$ to a constant and let
$p(\av) = \mathcal{N}(\tilde{\av},\varsigma^2 \Iv)$. Making use of
the E step already derived in Section \ref{sec:mean_field}, we only
need to complete the M step.

Utilizing Lemma \ref{lemma:schur_trace1} and ignoring the terms
independent of $\sigma^2$ and $\av$, we have \BEQ\BA{rl}
    &\mathcal{F}_{disc}(\mv_1,\cdots,\mv_T,\sigma^2,\av)\\
    =& \sum_{t=1}^T\left\{ \frac{N}{2}\log(\sigma^2) + \frac{1}{2\sigma^2}\av^T[\Mv_t^T\Sv^T\Sv\Mv_t - \diag(\Mv_t^T\Sv^T\Sv\Mv_t)]\av  + \frac{1}{2\sigma^2}
    \av^T \diag(\Sv^T\Sv)\av - \frac{1}{\sigma^2} \rv_t^T\Sv\Mv_t\av \right\} \\
    & +
    \frac{1}{2\varsigma^2}(\av-\tilde{\av})^T(\av-\tilde{\av}),
\EA\EEQ
where $\Mv_t = \diag(\mv_t)$. Equating $\partial
\mathcal{F}/\partial \av = \zerov$ produces \BEQ
    \textstyle \hat{\av} = \left\{\sum_{t=1}^T \diag(\Sv^T\Sv) + \Mv_t^T[\Sv^T\Sv-\diag(\Sv^T\Sv)]\Mv_t\right\}^{-1} \left( \sum_{t=1}^T
    \Mv_t\Sv^T\rv_t + \frac{\sigma^2}{\varsigma^2}\tilde{\av}\right).
\EEQ Substituting $\av = \hat{\av}$ into $\mathcal{F}$ and solve
for $\mathcal{F}(\mv_1,\cdots,\mv_T,\sigma^2,\hat{\av})/\partial
\sigma^{-2} = 0$ gives \BEQ\label{eq:sigma_hat1}
    \hat{\sigma}^2 = \frac{1}{NT} \left\{\sum_{t=1}^T\left[ (\rv_t-\Sv\hat{\Av}\mv_t)^T (\rv_t-\Sv\hat{\Av}\mv_t) + \sum_{k=1}^K(1-m_{t,k}^2)\hat{A}_k^2\cdot\sv_k^T\sv \right] \right\}.%
\EEQ

The variational EM algorithm for flooding-discrete-SISO MUD is
presented in Table \ref{tab:disc_siso_em}.

\begin{table}\caption{Variational EM algorithm employing Discrete SISO MUD.}\label{tab:disc_siso_em}
\centering \vspace{3pt}
\setlength{\tabcolsep}{3pt} \small
\begin{tabular}{|c|l|} \hline

\emph{Initialization} & \BT{c} Set $\tilde{\mv}_t = \zerov$,
$\LLR_{dec}(b_{t,k} = 0)$, $\sigma^{2(0)} = 0$, and $\av^{(0)} = \tilde{\av}$. \vspace{5pt} \ET\\

%%%%%%% The E Step
\BT{c} \emph{Update for $Q(\bv)$}\\ \emph{E Step} \ET &
\BT{lllllcl}
    FOR& \multicolumn{6}{l}{ $j = 1:J$ (\emph{Outer Iteration})}\\
    &\multicolumn{6}{l}{FOR $i = 1:I$ (\emph{Inner Iteration})} \\
    &   &\multicolumn{5}{l}{FOR $k=1:K$ and $t=1:T$} \\
    &   & & \multicolumn{4}{l}{$\LLR_{pos}(b_{t,k}) = \LLR_{dec}(b_{t,k}) + \frac{2}{\sigma^2}\left[\etav_k^T \rv_t - \betav_k^T\mv_t\right]$} \\
    &   & & \multicolumn{4}{l}{$m_{t,k} = \tanh[\LLR_{pos}(b_{t,k})/2]$} \\
    &   &END&&&&\\
    &END&&&&&\\
\ET\\

%%%%%%% The Decoding
\BT{c} \emph{BCJR Decoding} \ET & \hspace{18pt}
\BT{lllcl}
    &\multicolumn{4}{l}{FOR $k=1:K$ and $t=1:T$} \\
    &   & \multicolumn{3}{l}{$\LLR_{mud}(b_{t,k}) = \LLR_{pos}(b_{t,k}) - \LLR_{dec}(b_{t,k})$} \\
    &   & \multicolumn{3}{l}{$\LLR_{dec}(b_{t,k}) \stackrel{\scriptsize \mbox{Decoding}}{\Longleftarrow} \LLR_{mud}(b_{t,k})$} \hfill (\emph{Extrinsic Information})\\
    &   & $\hat{b}_{t,k}$ &=& $\tanh[\LLR_{mud}(b_{t,k})/2 + \LLR_{dec}(b_{t,k})/2 ]$ \hfill (\emph{Posterior Estimate})\\
    &   & $m_{t,k}$ &$\leftarrow$& $\hat{b}_{t,k}$\\
    &END& & &\\
\ET\\

%%%%%%% The M Step
\BT{c} \emph{Update for $\sigma^2$ and $\av$}\\ \emph{M Step} \ET &
\BT{lllcl}
    &   & $\av$ &$\leftarrow$& $\left\{\sum_{t=1}^T \Sv^T\Sv + \Mv_t^T[\Sv^T\Sv-\diag(\Sv^T\Sv)]\Mv_t + \frac{\sigma^2}{\varsigma^2}\Iv \right\}^{-1} \left( \sum_{t=1}^T \Mv_t\Sv^T\rv_t
    + \frac{\sigma^2}{\varsigma^2}\tilde{\av} \right)$\\
    &   & $\sigma^2$ &$\leftarrow$& $\frac{1}{NT} \left\{\sum_{t=1}^T\left[ (\rv_t-\Sv\Av\mv_t)^T (\rv_t-\Sv\Av\mv_t) - \sum_{k=1}^K (1-m_{t,k}^2)A_k^2\cdot \sv_k^T\sv_k \right]
    \right\}$\\
\ET\\
& END\\
\hline
\end{tabular}
\setlength{\tabcolsep}{5pt} \normalsize
\end{table}

Now consider taking a step backward and assume $\av$ to be known
perfectly, and only $\sigma^2$ needs to be estimated. It is clear
that, in (\ref{eq:sigma_hat1}), each element of $\mv_t$ converges
to $-1$ or $+1$ as the algorithm converges. Hence, the last term
in (\ref{eq:sigma_hat1}) eventually vanishes. By omitting the
vanishing term, this is exactly the equation to estimate
$\sigma_{tot}^2 = \sigma^2 + \sigma_{MU}^2$ in \cite{Alex98}.
Together with Section \ref{sec:mean_field}, we have now completed
the interpretation of Alexander-Grant-Reed's turbo detector as an
instance of the variational EM algorithm.

%Also note that after the inserted decoding operation, the M Steps in
%flooding-Gaussian-SISO MUD and flooding-discrete-SISO MUD are
%identical, despite rather distinct expressions.

%%%%%%%%%%%%%%%%%%%%%%%%%%%%%%%%%%%%%%%%%%%%%%%%%%
\section{Simulation Results}\label{sec:simulation1}

In this section, we investigate the performance of turbo multiuser
detectors employing the Gaussian SISO MUD and DDF-aided discrete
SISO MUD (the original-form discrete SISO MUD would suffer from
poor convergence due to the non-convexity of
$\mathcal{F}_{disc}(\mv)$). We will consider two different
scenarios to test the proposed schemes in both standard benchmark
settings and in more practical situations:

\begin{description}
\item[Scenario I:] A flat-fading
synchronous CDMA channel same as that in \cite{Wangx99}: A four-user
system is assumed with equal cross-correlation $\rho = 0.7$. All
users have equal power and employ the same rate $1/2$ convolutional
code with generators $10011$ and $11101$.

\item[Scenario II:] A flat-fading
synchronous CDMA channel with random spreading: The system has
spreading gain of $N=32$ and $K=32$ active users. All users have
equal power and employ the same rate $1/2$ convolutional code with
generators $111$ and $101$. In this case, we also assume the
receiver has noisy channel information (unknown $\sigma^2$ and
inaccurate channel amplitude estimates).
\end{description}

Since the focus of this paper is the introduction of a theoretical
framework, these test results are for proof-of-concept purposes
only and are by no means complete. More elaborate and complete
simulations, such as the near-far situation or the extensions to
multipath channels, may be performed following the same recipe
presented in \cite{Wangx99} and will be omitted.

\subsection{Gaussian SISO MUD}

\subsubsection{Scenario I with Perfect Channel Knowledge}

In Section \ref{sec:stand_gaussian_siso}, we proved that the
Wang-Poor turbo MUD scheme is an instance of hybrid-Gaussian-SISO
MUD, whose performance is depicted in \fig 3 of \cite{Wangx99}. In
this paper, we will realize the other two variations, namely
sequential-Gaussian-SISO MUD and flooding-Gaussian-SISO MUD,
predicted by the theory of VFEM and the associated message passing
rules. The complete Gaussian SISO MUD algorithms are described in
Table \ref{tab:gauss_siso}.

\begin{figure}
\begin{center}
\scalebox{0.9}{\includegraphics{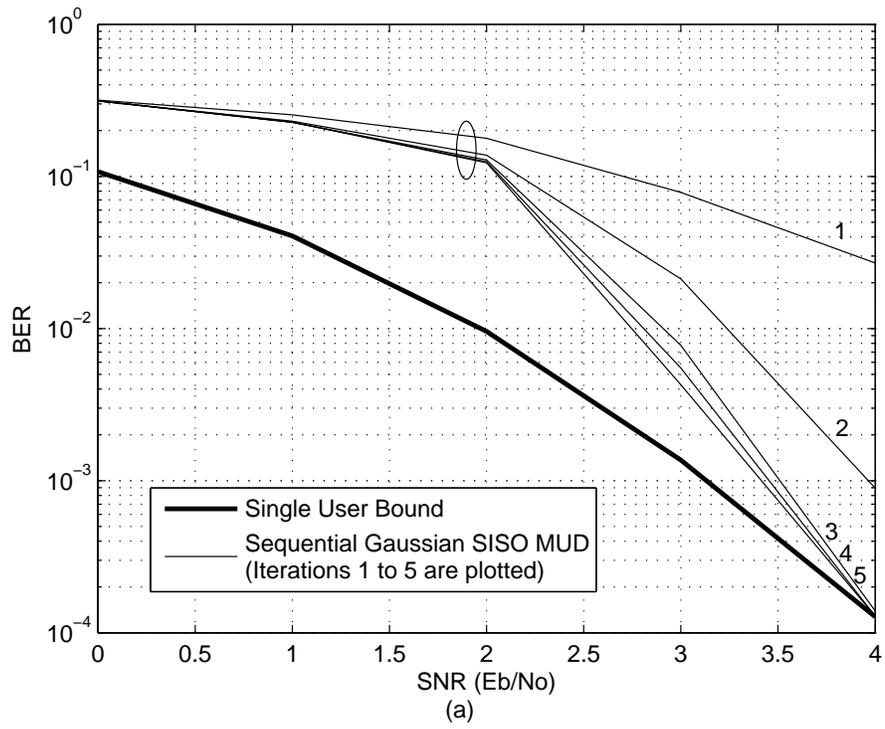}}\\%
\scalebox{0.9}{\includegraphics{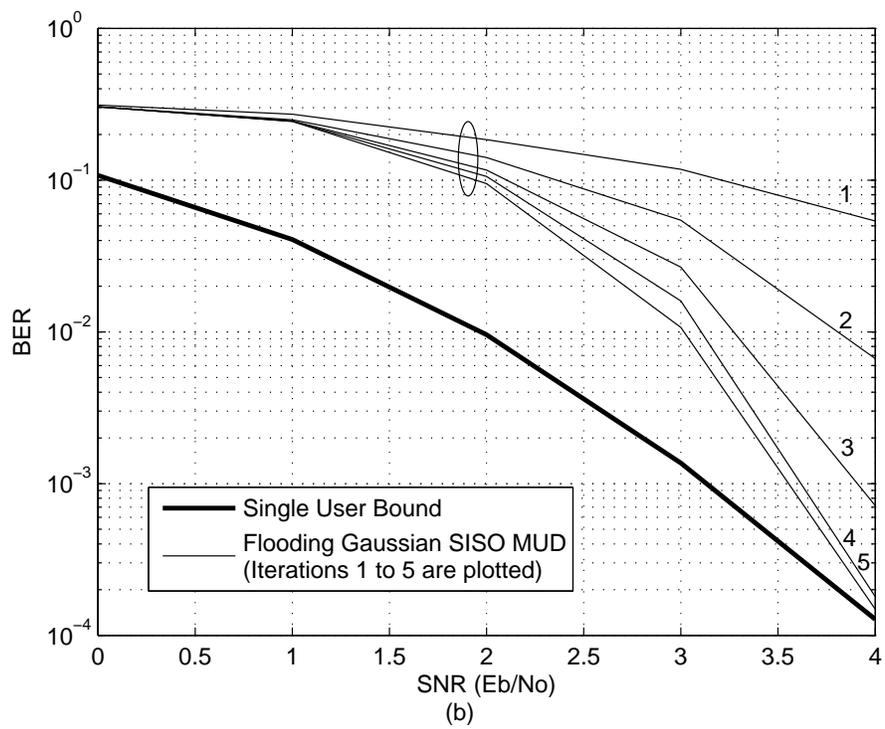}}%
\caption{BER performance of turbo MUD employing Gaussian-SISO MUD
($K=4$, $\rho=0.7$). (a) Sequential schedule; (b) Flooding
schedule.}\label{fig:bpsk_gauss}
\end{center}
\end{figure}

\fig \ref{fig:bpsk_gauss}(a) and \fig \ref{fig:bpsk_gauss}(b) plot
the BER performance of turbo MUD employing sequential-Gaussian-SISO
MUD and flooding-Gaussian-SISO MUD, respectively, in simulation
scenario I. The results after each of the $J=5$ outer iterations are
recorded. It is shown that both schemes out-perform
hybrid-Gaussian-SISO MUD, which was originally proposed outside of
the variational inference framework. The BER improvement, although
small, verifies that the sequential and flooding scheduling schemes
are sound in practice as they are in theory.

In the above simulations, we find that the difference in performance
between sequential schedule and flooding schedule is small. For
conciseness, in the case of inaccurate channel estimates, we will
only consider the flooding schedule, since it in general leads to
lower overall complexity and latency at both the detection and
decoding stages. That been said, one may choose to implement the
sequential or hybrid schedule with ease as special need arises.

\subsubsection{Scenario II with Unknown Noise Variance and Inaccurate Channel
Estimates}

We now consider simulation scenario II, a more challenging
situation where the receiver is assumed to have no noise variance
information and only inaccurate channel estimates. The actual
channel is assumed to be Gaussian-distributed conditioned on the
inaccurate estimate $\tilde{\av}$, \ie, $p(\av) =
\mathcal{N}(\tilde{\av}, \varsigma^2 \Iv)$, as in Section
\ref{sec:em_sisommse}. In the simulations, we fix the true channel
$A_k =1$ (or $\av = \onev$), and generate the noisy channel
estimate $\tilde{A}_k$ as $1 + \delta_k$, where $p(\delta_k) =
\mathcal{N}(0,\varsigma^2)$.

\begin{figure}

\hspace{-0.5in}\scalebox{1.0}{\includegraphics{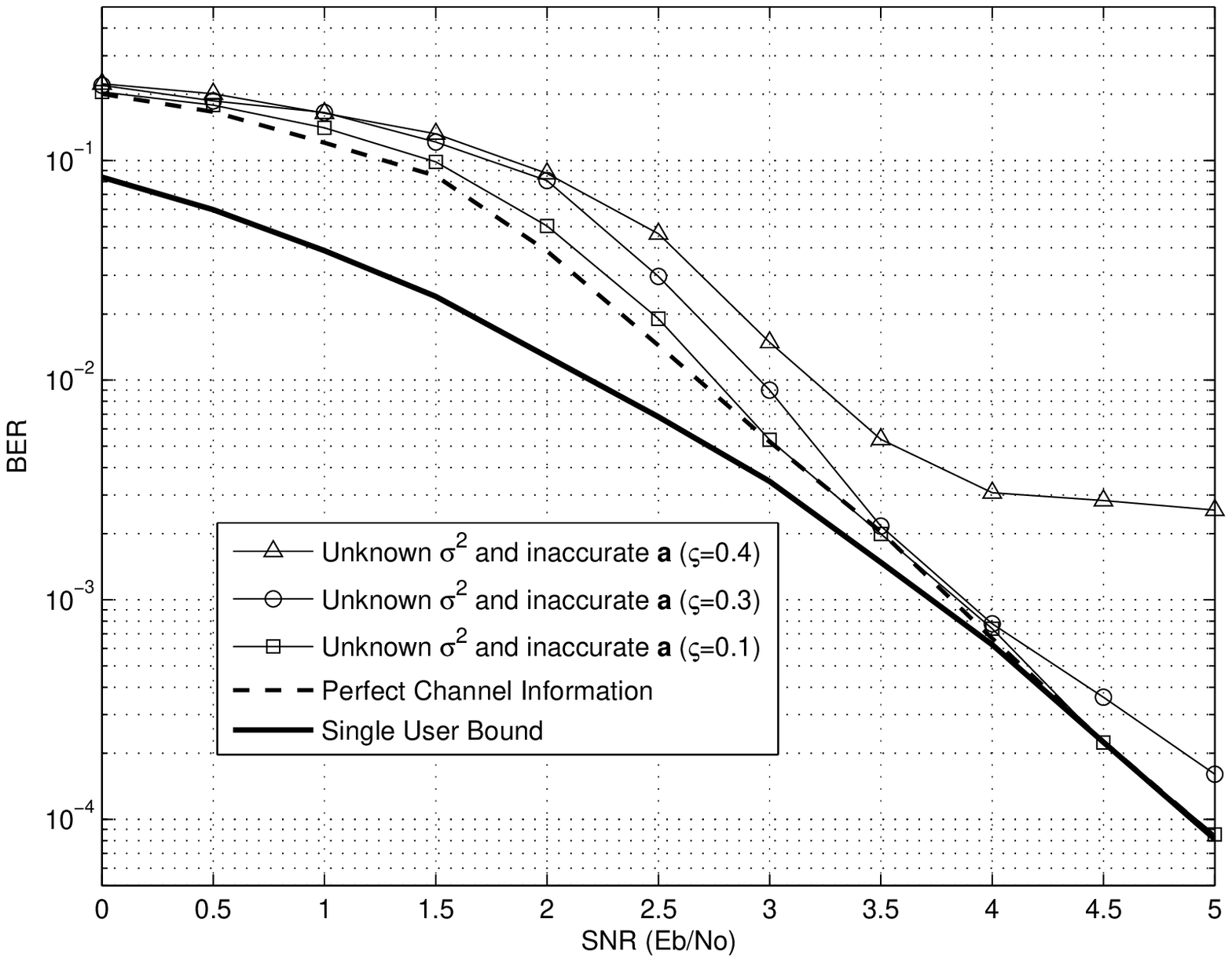}}\\%
\caption{BER performance of turbo MUD employing
flooding-Gaussian-SISO MUD with joint noise variance and channel
estimation ($N=32$, $K=32$). The single user bound is obtained by
assuming perfect channel knowledge.}\label{fig:gauss_em}

\end{figure}

\fig \ref{fig:gauss_em} depicts the flooding-Gaussian-SISO MUD
implemented with joint $\av$ and $\sigma^2$ estimation. We set
$\varsigma$ to be $0.1$, $0.3$ and $0.4$, respectively, to be
compared with the case of exact channel knowledge at the receiver.
To be consistent, the curves plotted are the results after $J=10$
outer iterations, although in most cases convergence is achieved
with fewer iterations. It is seen that, with the help of the
variational EM algorithm, this turbo multiuser detector is very
robust to severe channel estimation error, up to $\varsigma = 0.3$.
It is only when $\varsigma$ reaches $0.4$, \ie, $40\%$ that of the
actual channel amplitude, significant performance degradation starts
to appear.

\subsection{DDF-Aided Discrete SISO MUD}

\subsubsection{Scenario I with Unknown Noise Variance Only}

To implement DDF-aided discrete SISO MUD in turbo MUD, we simply
need to replace the first outer iteration ($j=1$) in the algorithms
of Table \ref{tab:disc_siso} with DDF update
(\ref{eq:mk_update_DDF4}), and keep the remaining outer iterations
($j\geq 2$) the same. We will first simulate scenario I with turbo
MUD employing DDF-aided sequential-discrete-SISO MUD and DDF-aided
flooding-discrete-SISO MUD, each having $I=6$ inner iterations
within every outer iteration, except for the first outer iteration,
where the DDF update only requires $1$ inner iteration. In the
detection algorithms prescribed in Table \ref{tab:disc_siso}, we
added a noise variance estimate step like it is done in
\cite{Alex98,Koba01}. This is a special case of the variational EM
algorithm introduced in Section \ref{sec:em_meanfield} with
$\sigma^2$ alone being the unknown parameter. Having the noise
variance as an unknown does not seem to temper the detector
performance compared to perfectly-known noise variance, and, in
certain cases, even helps. We attribute this phenomenon to its
possible relation with optimizing (\ref{eq:f_meanfield}) via
\emph{simulation annealing} by setting $\sigma^2$ to be the
\emph{temperature} parameter \cite{Mackay98}, but will leave
detailed discussions to future works.

\begin{figure}
\begin{center}
\scalebox{0.9}{\includegraphics{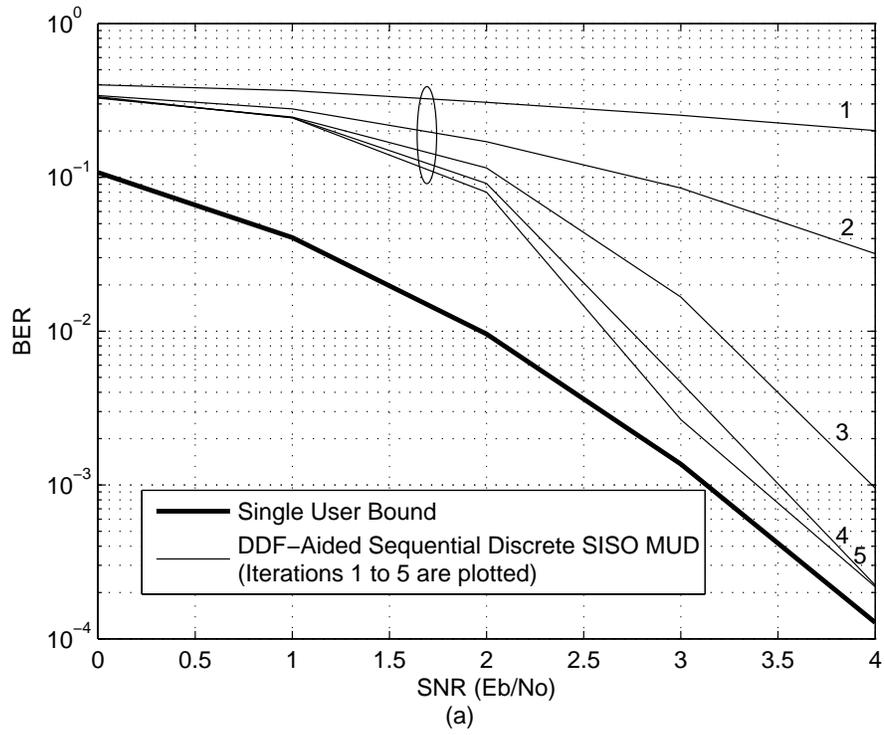}}\\%
\scalebox{0.9}{\includegraphics{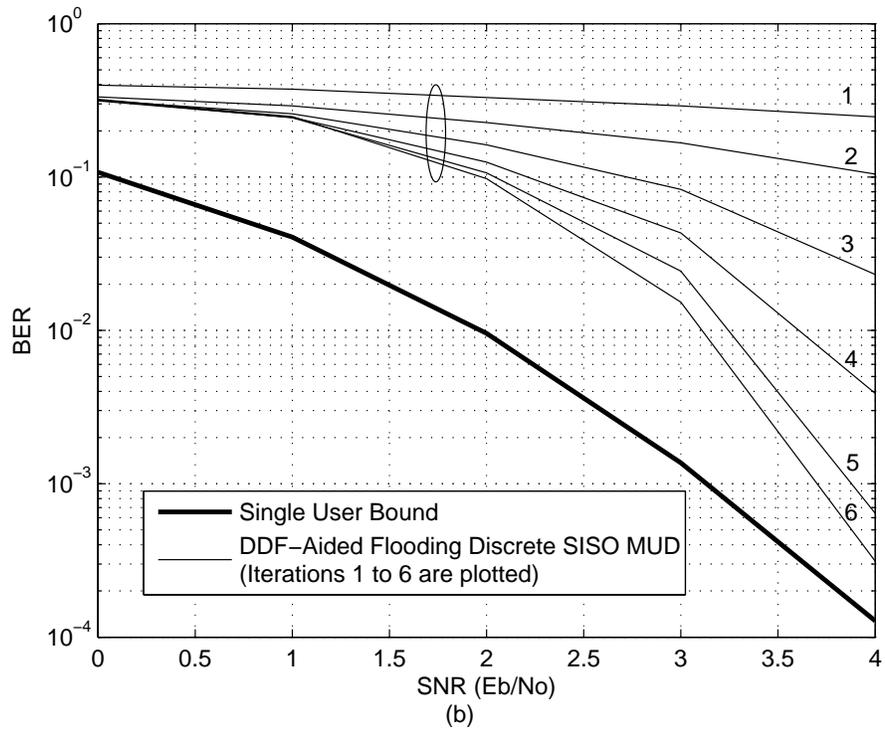}}%
\caption{BER performance of turbo MUD employing discrete-SISO MUD
with joint noise variance estimation ($K=4$, $\rho=0.7$). (a)
Sequential schedule; (b) Flooding schedule.}\label{fig:bpsk_disc}
\end{center}
\end{figure}

\fig \ref{fig:bpsk_disc}(a) and \fig \ref{fig:bpsk_disc}(b) depict
the BER performance of the above-mentioned schemes over $J=5$ or $6$
outer iteration. Despite the slightly inferior performance compared
to the Gaussian SISO counterparts, the DDF-aided discrete SISO
detectors have been shown to produce excellent results even with
unknown noise variance. The existing discrete SISO detectors, such
as \cite{Alex98,Koba01}, would fail under such simulation settings,
due to the lack of DDF initialization.

%The issues of iterative noise-plus-MAI variance and channel
%amplitude estimation discussed in \cite{Alex98,Koba01} will be
%postponed until Section \ref{sec:var_EM}.
%
%In the above simulations on both Gaussian SISO MUD and modified
%discrete SISO MUD, we find that the difference in performance
%between sequential schedule and flooding schedule is small. For
%conciseness, in the rest of the paper we will only consider the
%flooding schedule in our study, since it in general leads to lower
%overall complexity and latency at both the detection and decoding
%stage. That been said, one may choose to implement the sequential or
%hybrid schedule as need arises with only minor changes to the
%algorithm.

\subsubsection{Scenario II with Unknown Noise Variance and Inaccurate Channel
Estimates}

We now further investigate the case of inaccurate channel estimates
in addition to unknown noise variance with simulation scenario II.
Like the Gaussian SISO case, we set the true channel to be $A_k =
1$, and generate noisy channel estimates by letting $\tilde{A}_k$ be
$1 + \delta_k$, where $p(\delta_k) = \mathcal{N}(0,\varsigma^2)$.

\begin{figure}

\hspace{-0.5in}\scalebox{1.0}{\includegraphics{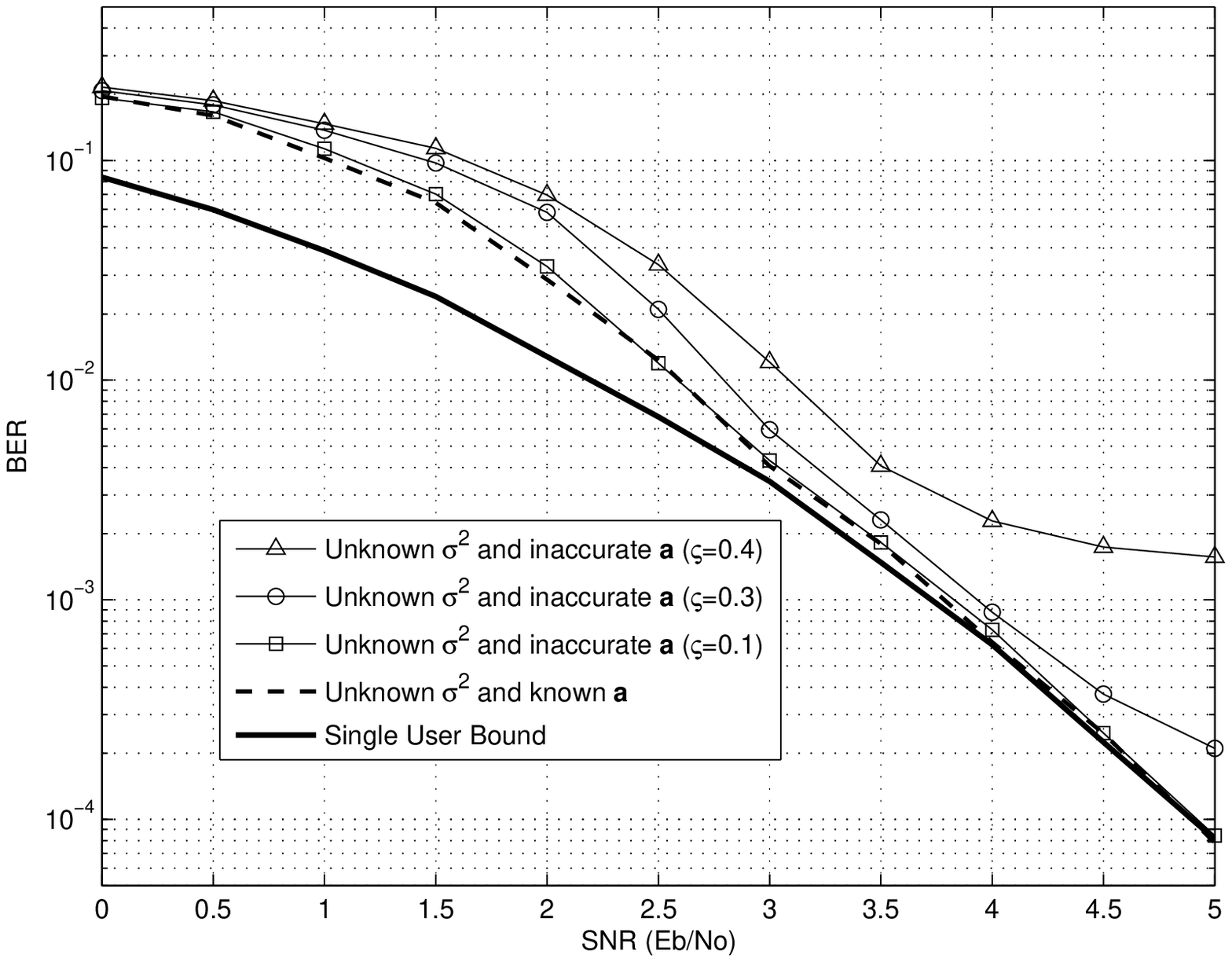}}\\%
\caption{BER performance of turbo MUD employing
flooding-discrete-SISO MUD with joint noise variance and channel
estimation ($N=32$, $K=32$). The single user bound is obtained by
assuming perfect channel knowledge.}\label{fig:disc_em}

\end{figure}

\fig \ref{fig:disc_em} depicts the performance of DDF-aided
flooding-discrete-SISO MUD under channel estimation error of
$\varsigma = 0.1$, $0.3$ and $0.4$, respectively. It is shown that,
by iteratively refining the channel estimates, the modified
flooding-discrete-SISO MUD is also robust to significant errors in
channel estimation.

%%%%%%%%%%%%%%%%%%%%%%%%%%%%%%%%%%%%%%%%%%%%%%%%%%%%%%%
\section{Conclusions}\label{sec:conclusion}

The concept of free energy is a far-reaching one. Besides its
original formulation in statistical physics, it also recently
finds its application in interpreting various probabilistic
inference techniques, such as the EM algorithm \cite{Neal98} and
belief propagation (sum-product algorithm) \cite{Yedi05}.

The main focus of this paper is the introduction of a
comprehensive theory, centered around the minimization of
variational free energy, that would describe various SISO detectors
in multiple access channels. In particular, we developed guidelines for SISO
detector design in linear Gaussian vector channels, first by
pointing out the importance of message-passing scheduling, and
next by deriving detection algorithms accordingly. We show that it
is a carefully-chosen scheduling scheme combined with its
accompanying SISO detector that produces a good turbo detector,
opposed to the conventional view that focuses on the detector
design alone. With this new paradigm comes a spectrum of plausible
SISO detectors. In addition to new detectors constructed as a
result, we show that some of the classic SISO detectors can be
seen as special instances of this composition, and subsequently
refined systematically.

In the algorithm design stage, after the \emph{postulation} of
prior and posterior distributions with the help of some intuition, it may
be seen that our efforts in obtaining good algorithms have been
condensed to the \emph{evaluation} and \emph{optimization} of free
energy expressions, such as the ones at the centre of this paper,
$\mathcal{F}_{gauss}(\muv,\Sigmav)$ and $\mathcal{F}_{disc}(\mv)$.
By viewing existing MUD schemes under the same roof, we obtain
many interesting insights. Furthermore, we extended the initiative
to variational-EM-based MUD, in which channel parameters may be
efficiently and blindly estimated in conjunction with turbo MUD.

%-------------------------------------------------------------------------

\bibliographystyle{IEEEtran}
\small\bibliography{../../../subspace}

\begin{thebibliography}{10}
\providecommand{\url}[1]{#1}
\csname url@rmstyle\endcsname
\providecommand{\newblock}{\relax}
\providecommand{\bibinfo}[2]{#2}
\providecommand\BIBentrySTDinterwordspacing{\spaceskip=0pt\relax}
\providecommand\BIBentryALTinterwordstretchfactor{4}
\providecommand\BIBentryALTinterwordspacing{\spaceskip=\fontdimen2\font plus
\BIBentryALTinterwordstretchfactor\fontdimen3\font minus
  \fontdimen4\font\relax}
\providecommand\BIBforeignlanguage[2]{{%
\expandafter\ifx\csname l@#1\endcsname\relax
\typeout{** WARNING: IEEEtran.bst: No hyphenation pattern has been}%
\typeout{** loaded for the language `#1'. Using the pattern for}%
\typeout{** the default language instead.}%
\else
\language=\csname l@#1\endcsname
\fi
#2}}

\bibitem{Berrou93}
C.~Berrou, A.~Glavieux, and P.~Titmajshima, ``Near {Shannon} limit
  error-correction coding and decoding: Turbo codes,'' in \emph{Proc. IEEE Int.
  Conf. Commun. (ICC'93)}, May 1993, pp. 1064--1070.

\bibitem{Alex98}
P.~D. Alexander, A.~J. Grant, and M.~C. Reed, ``Iterative detection in
  code-division multiple-access with error control coding,'' \emph{Euro. Trans.
  Telecommun.}, vol.~9, no.~5, pp. 419--426, Oct. 1998.

\bibitem{Wangx99}
X.~Wang and H.~V. Poor, ``Iterative (turbo) soft interference cancellation and
  decoding for coded {CDMA},'' \emph{IEEE Trans. Commun.}, vol.~47, no.~7, pp.
  1046--1061, July 1999.

\bibitem{Doui95}
C.~Douillard, A.~Picart, P.~Didier, M.~J\'{e}z\'{e}quel, C.~Berrou, and
  A.~Glavieux, ``Iterative correction of intersymbol interference: Turbo
  equalization,'' \emph{Euro. Trans. Telecommun.}, vol.~6, pp. 507--511, Sept.
  1995.

\bibitem{Tuch02a}
M.~T\"{u}chler, A.~C. Singer, and R.~Koetter, ``Turbo equalization: Principles
  and new results,'' \emph{IEEE Trans. Commun.}, vol.~50, no.~5, pp. 754--767,
  May 2002.

\bibitem{Ariya00}
S.~L. Ariyavisitakul, ``Turbo space-time processing to improve wireless channel
  capacity,'' \emph{IEEE Trans. Commun.}, vol.~48, no.~8, pp. 1347--1359, Aug.
  2000.

\bibitem{Haykin04}
S.~Haykin, M.~Sellathurai, Y.~de~Jong, and T.~Willink, ``Turbo-{MIMO} for
  wireless communications,'' \emph{IEEE Commun. Mag.}, vol.~42, no.~10, pp.
  48--53, Oct. 2004.

\bibitem{Verdu98}
S.~Verd\`{u}, \emph{Multiuser Detection}.\hskip 1em plus 0.5em minus
  0.4em\relax Cambridge, U.K.: Cambridge University Press, 1998.

\bibitem{Bahl74}
L.~Bahl, J.~Cocke, F.~Jelinek, and J.~Raviv, ``Optimal decoding of linear codes
  for minimizing symbol error rate,'' \emph{IEEE Trans. Inf. Theory}, vol.
  IT-20, pp. 284--287, Mar. 1974.

\bibitem{Mackay03}
D.~MacKay, \emph{Information Theory, Inference and Learning Algorithms}.\hskip
  1em plus 0.5em minus 0.4em\relax Cambridge U.K.: Cambridge University Press,
  2003.

\bibitem{Ksch01}
F.~R. Kschischang, B.~J. Frey, and H.~A. Loeliger, ``Factor graphs and the
  sum-product algorithm,'' \emph{IEEE Trans. Inf. Theory}, vol.~47, no.~2, pp.
  498--519, Feb. 2001.

\bibitem{Bout02}
J.~Boutros and G.~Caire, ``Iterative multiuser joint decoding: Unified
  framework and asymptotic analysis,'' \emph{IEEE Trans. Inf. Theory}, vol.~48,
  no.~7, pp. 1772--1793, July 2002.

\bibitem{Tanaka02}
T.~Tanaka, ``A statistical-mechanics approach to large-system analysis of
  {CDMA} multiuser detectors,'' \emph{IEEE Trans. Inf. Theory}, vol.~48,
  no.~11, pp. 2888--2910, Nov. 2002.

\bibitem{Guo05}
D.~Guo and S.~Verd\'{u}, ``Randomly spread {CDMA}: {Asymptotics} via
  statistical physics,'' \emph{IEEE Trans. Inf. Theory}, vol.~51, no.~6, pp.
  1983--2010, June 2005.

\bibitem{Neal98}
R.~M. Neal and G.~E. Hinton, ``A view of the {EM} algorithm that justifies
  incremental, sparse, and other variants,'' in \emph{Learning in Graphical
  Models}, M.~I. Jordan, Ed.\hskip 1em plus 0.5em minus 0.4em\relax Kluwer
  Academic Publishers, 1998.

\bibitem{Lin06isit}
D.~D. Lin and T.~J. Lim, ``Multiuser detection of {$M$-QAM} symbols via
  bit-level equalization and soft detection,'' in \emph{Proc. IEEE Int. Symp.
  Information Theory (ISIT'06)}, July 2006, pp. 1909--1913.

\bibitem{Lin07isit}
------, ``Turbo equalization for {Gray}-coded m-ary {QAM} with bit-level soft
  decisions,'' in \emph{Proc. IEEE Int. Symp. Information Theory (ISIT'07)},
  June 2007, pp. 1701--1705.

\bibitem{Ksch98}
F.~R. Kschischang and B.~J. Frey, ``Iterative decoding of compound codes by
  probability propagation in graphical models,'' \emph{IEEE J. Sel. Areas
  Commun.}, vol.~16, no.~2, pp. 219--230, Feb. 1998.

\bibitem{Yacov07}
N.~Yacov, H.~Efraim, H.~Kfir, I.~Kanter, and O.~Shental, ``Parallel vs.
  sequential belief propagation decoding of {LDPC} codes over ${GF}(q)$ and
  {Markov} sources,'' \emph{Physica A: Statistical Mechanics and Its
  Applications}, vol. 378, no.~2, pp. 329--335, May 2007.

\bibitem{Jordan04}
M.~I. Jordan, \emph{An Introduction to Probabilistic Graphical Models}.\hskip
  1em plus 0.5em minus 0.4em\relax In Preparation.

\bibitem{Frey98}
B.~J. Frey, \emph{Graphical Models for Machine Learning and Digital
  Communication}.\hskip 1em plus 0.5em minus 0.4em\relax Cambridge,
  Massachusetts: MIT Press, 1998.

\bibitem{Jordan99}
M.~I. Jordan, Z.~Ghahramani, T.~Jaakkola, and L.~K. Saul, ``An introduction to
  variational methods for graphical models,'' \emph{Machine Learning}, vol.~37,
  no.~2, pp. 183--233, 1999.

\bibitem{Path04}
R.~K. Pathria, \emph{Statistical Mechanics}.\hskip 1em plus 0.5em minus
  0.4em\relax Oxford: Elsevier Butterworth-Heinemann, 2004.

\bibitem{Elder98}
H.~Elders-Boll, H.~D. Schotten, and A.~Busboom, ``Efficient implementation of
  linear multiuser detectors for asynchronous {CDMA} systems by linear
  interference cancellation,'' \emph{Euro. Trans. Telecommun.}, vol.~9, no.~5,
  pp. 427--438, Oct. 1998.

\bibitem{Ras00}
L.~K. Rasmussen, T.~J. Lim, and A.~L. Johansson, ``A matrix-algebraic approach
  to successive interference cancellation in {CDMA},'' \emph{IEEE Trans.
  Commun.}, vol.~48, no.~1, pp. 145--151, Jan. 2000.

\bibitem{Guo00}
D.~Guo, L.~K. Rasmussen, S.~Sun, and T.~J. Lim, ``A matrix-algebraic approach
  to linear parallel interference cancellation in {CDMA},'' \emph{IEEE Trans.
  Commun.}, vol.~48, no.~1, pp. 152--161, Jan. 2000.

\bibitem{Tan01}
P.~H. Tan, L.~K. Rasmussen, and T.~J. Lim, ``Constrained maximum-likelihood
  detection in {CDMA},'' \emph{IEEE Trans. Commun.}, vol.~49, no.~1, pp.
  142--153, Jan. 2001.

\bibitem{Luo92}
Z.~Q. Luo and P.~Tseng, ``On the convergence of the coordinate descent method
  for convex differentiable minimization,'' \emph{J. Optim. Theory Appl.},
  vol.~72, no.~1, pp. 7--35, Jan. 1992.

\bibitem{Fabri02}
T.~Fabricius and O.~N{\o}klit, ``Approximations to joint-{ML} and {ML}
  symbol-channel estimators in {MUD CDMA},'' in \emph{Proc. IEEE Global
  Telecommunications Conf. (Globecom'02)}, vol.~1, Nov. 2002, pp. 389--393.

\bibitem{Kaba03}
Y.~Kabashima, ``A {CDMA} multiuser detection algorithm on the basis of belief
  propagation,'' \emph{J. Phys. A: Math. Gen.}, vol.~36, pp. 11\,111--11\,121,
  Oct. 2003.

\bibitem{Koba01}
M.~Kobayashi, J.~Boutros, and G.~Caire, ``Successive interference cancellation
  with {SISO} decoding and {EM} channel estimation,'' \emph{IEEE J. Sel. Areas
  Commun.}, vol.~19, no.~8, pp. 1450--1460, Aug. 2001.

\bibitem{Duel93}
A.~Duel-Hallen, ``Decorrelating decision-feedback multiuser detector for
  synchronous code-division multiple-access channel,'' \emph{IEEE Trans.
  Commun.}, vol.~41, no.~2, pp. 285--290, Feb. 1993.

\bibitem{Evans00}
J.~Evans and D.~N.~C. Tse, ``Large system performance of linear multiuser
  receivers in multipath fading channels,'' \emph{IEEE Trans. Inf. Theory},
  vol.~46, no.~6, pp. 2059--2078, Sept. 2000.

\bibitem{Li04}
H.~Li and H.~V. Poor, ``Impact of imperfect channel estimation on turbo
  multiuser detectiion in {DS-CDMA} systems,'' in \emph{Proc. IEEE Wireless
  Communications Networking Conf. (WCNC'04)}, vol.~1, Mar. 2004, pp. 30--35.

\bibitem{Li05}
------, ``Impact of channel estimation errors on multiuser detection via the
  replica method,'' \emph{EURASIP J. Wireless Commun. Networking}, vol. 2005,
  no.~2, pp. 175--186, 2005.

\bibitem{Demp77}
A.~P. Dempster, N.~M. Laird, and D.~B. Rubin, ``Maximum likelihood from
  incomplete data via the {EM} algorithm,'' \emph{J. Roy. Stat. Soc.}, vol.~39,
  no.~1, pp. 1--38, 1977.

\bibitem{Kocian03}
A.~Kocian and B.~H. Fleury, ``{EM}-based joint data detection and channel
  estimation of {DS-CDMA} signals,'' \emph{IEEE Trans. Commun.}, vol.~51,
  no.~10, pp. 1709--1720, Oct. 2003.

\bibitem{Frey05}
B.~J. Frey and N.~Jojic, ``A comparison of algorithms for inference and
  learning in probabilistic graphical models,'' \emph{IEEE Trans. Patt. Anal.
  Mach. Int.}, vol.~27, no.~9, pp. 1392--1416, Sept. 2005.

\bibitem{Lin05detect}
D.~D. Lin and T.~J. Lim, ``The variational inference approach to joint data
  detection and phase noise estimation in {OFDM},'' \emph{IEEE Trans. Signal
  Processing}, vol.~55, no.~5, pp. 1862--1874, May 2007.

\bibitem{Niss06}
M.~Nissil\"{a} and S.~Pasupathy, ``Joint estimation of carrier frequency offset
  and statistical parameters of the multipath fading channel,'' \emph{IEEE
  Trans. Commun.}, vol.~54, no.~6, pp. 1038--1048, June 2006.

\bibitem{Medard00}
M.~M\'{e}dard, ``The effect upon channel capacity in wireless communications of
  perfect and imperfect knowledge of the channel,'' \emph{IEEE Trans. Inf.
  Theory}, vol.~46, no.~3, pp. 933--946, May 2000.

\bibitem{Mackay98}
D.~Mackay, ``Introduction to {Monte Carlo} methods,'' in \emph{Learning in
  Graphical Models}, M.~I. Jordan, Ed.\hskip 1em plus 0.5em minus 0.4em\relax
  Kluwer Academic Publishers, 1998.

\bibitem{Yedi05}
J.~S. Yedidia, W.~T. Freeman, and Y.~Weiss, ``Constructing free-energy
  approximations and generalized belief propagation algorithms,'' \emph{IEEE
  Trans. Inf. Theory}, vol.~51, no.~7, pp. 2282--2312, July 2005.

\end{thebibliography}
\end{document}